\documentclass[sigconf]{acmart}

\renewcommand\footnotetextcopyrightpermission[1]{} 

\usepackage{booktabs} 
\usepackage{epsfig}
\usepackage{graphics} 
\usepackage{graphicx}
\usepackage[font=small,skip=0pt]{caption}
\usepackage[font=small,skip=0pt]{subcaption}
\usepackage{url}      
\usepackage{bm}
\usepackage{calc}
\usepackage{multirow}

\usepackage{amsmath,amssymb}
\usepackage{mathrsfs}
\usepackage{mathtools}
\usepackage{amsmath}
\usepackage{amsthm}
\usepackage{balance}
\usepackage{xcolor}

\usepackage{array}
\newcolumntype{L}[1]{>{\raggedright\let\newline\\\arraybackslash\hspace{0pt}}m{#1}}
\newcolumntype{C}[1]{>{\centering\let\newline\\\arraybackslash\hspace{0pt}}m{#1}}
\newcolumntype{R}[1]{>{\raggedleft\let\newline\\\arraybackslash\hspace{0pt}}m{#1}}

\usepackage[linesnumbered,ruled,vlined]{algorithm2e} 
\SetKwInOut{Parameter}{Parameter\hspace{-10pt}}

\newtheorem{theorem}{Theorem}

\newcommand{\ccomment}[1]{}

\newcommand\name{DeepRT}

\acmConference[SEC '21]{The Sixth ACM/IEEE Symposium on Edge Computing}{December 14--17, 2021}{San Jose, CA}

\begin{document}

\title{\name: A Soft Real Time Scheduler for Computer Vision Applications on the Edge}  

%
%
%

\ccomment{
\author{
    \IEEEauthorblockN{Zhe Yang, Klara Nahrstedt, Hongpeng Guo}
   \\\ \IEEEauthorblockA{University of Illinois at Urbana-Champaign  
    \\\{zheyang3, klara, hg5\}@illinois.edu}
}}
   \author{Zhe Yang, Klara Nahrstedt, Hongpeng Guo, Qian Zhou} 
    \affiliation{%
    \institution{Department of Computer Science}
      \country{University of Illinois at Urbana-Champaign}
    }
    \email{{zheyang3, klara, hg5, qianz}@illinois.edu}

\begin{abstract}
The ubiquity of smartphone cameras and IoT cameras, together with the recent boom of deep learning and deep neural networks, proliferate various computer vision driven mobile and IoT applications deployed on the edge. 
This paper focuses on applications which make soft real time requests to perform inference on their data -- they desire prompt responses within designated deadlines, but occasional deadline misses are acceptable. Supporting soft real time applications on a multi-tenant edge server is not easy, since the  requests  sharing the limited GPU computing resources of an edge server  interfere with each other. In order to 
tackle this problem,  we comprehensively evaluate how latency and throughput respond to different GPU execution plans. Based on this analysis, we propose a GPU scheduler, {\name}, which provides latency guarantee to the requests while maintaining high overall system throughput.  The key component of {\name}, DisBatcher, batches data from different requests as much as possible while it is proven to provide latency guarantee for requests admitted by an  Admission Control Module. {\name} also includes an Adaptation Module which tackles overruns. Our evaluation results show that {\name} outperforms state-of-the-art works in terms of the number of deadline misses and throughput.
\end{abstract}
\vspace{-1mm}

\maketitle
\thispagestyle{fancy}

\section{Introduction}
\label{sec:introduction}

\ccomment{
TODO \name is a layer between xxx and GPU?

x TODO say CNN.

x TODO keep in mind that latency = execution time + queuing delay

TODO mention execution model?

TODO if job instance execution time > period or deadline, reject

TODO the right question is not how many to batch, the right question when to batch

x TODO mention the difference between cloud inference and edge inference. Cloud inference focuses on throughput of a lot of requests. Edge inference focuses on reducing latency for a limited number of requests.

x TODO mention that ave latency of concurrent 2 latency higher than sequential 2

TODO how to call deadline pass part
}

The ubiquitous smartphones and Internet of Things (IoT) platforms, such as smart home solutions \cite{smartthings} and modern scientific experiment frameworks \cite{nguyen20174ceed}, produce a tremendous amount of data every day, especially video data. Meanwhile, we are also witnessing a rapid development of deep neural networks, especially Convolutional neural networks (CNNs),  and their hardware accelerators which empower fast and large-scale neural networks training and inference. These two trends proliferate a wide variety of computer vision applications across self-driving cars \cite{bojarski2016end}, mobile augmented reality \cite{jain2015overlay}\cite{liu2018edge}, mobile adaptive video streaming \cite{yeo2018neural}, and vehicle re-identification in urban surveillance \cite{liu2016deep}, 
to name a few. These applications benefit from CNN's excellent performance in making predictions through incorporating inference of trained CNN models into system design.

However, using the CNN models to build vision applications does not come for free. Smartphones or IoT  devices usually do not have sufficient computing or memory resources to support prompt CNN inference in place. While offloading CNN computations to cloud servers is an option \cite{crankshaw2017clipper}\cite{gujarati2017swayam}, many works propose to perform deep learning inference on the edge servers \cite{fang2019serving}\cite{hsu2019couper}\cite{zhou2019adaptive}. The reasons are twofold: (1)  Some data that need to be processed by CNN models, \emph{e.g.} images taken from smartphone cameras, contain private or proprietary information. Users are reluctant to upload them to a public cloud server for processing. (2) A lot of the aforementioned applications are sensitive to latency, and require real time CNN inference. For example,  an interactive application typically requires a response time less than $100$ milliseconds \cite{miller1968response}.    But the wide area network links between users and cloud servers exhibit the notorious issue of unbounded delay and jitter,  which could greatly undermine user experience. 

In this work, we limit our scope to handling these latency sensitive applications on the edge. Specifically, we focus on soft real time inference requests that desire real time responses but can tolerate occasional deadline misses, such as mobile augmented reality, mobile neural adaptive video streaming, and path planning in self-driving \cite{liu2020removing}. 
In order to guarantee real time services, we can certainly dedicate the deep learning accelerator on an edge server to a single application, but it's a waste of the precious accelerator resources. Supporting multitenancy and sharing the computation resources among multiple applications that have access to the edge decreases the cost of each client, but it inevitably affects latency performance of each application, as edge servers don't have unlimited resources due to space and budget limits. 


Faced with the trade-off between reducing application latency and increasing system throughput, a number of research approaches are proposed. 
In \cite{crankshaw2017clipper}, the authors propose a cloud based prediction serving system. It employs adaptive batching to maximize throughput while trying to reach a query latency target. However, cloud based solutions cannot be directly migrated to the edge paradigm. The reasons are twofold. First, cloud based solutions assume that the resources are abundant. Second, cloud servers aim to provide services for plenty of users so they usually set throughput as their primary goal, whereas clients seeking edge based services are most concerned about low latency guarantee.
In  \cite{fang2017qos}, an adaptive batching algorithm is proposed to increase GPU utilization which both increases system throughput and decreases average latency of the tasks, but it does not target soft real time inference requests. DeepQuery \cite{fang2019serving} considers real time tasks, but its major focus is to optimize the non-real time tasks while  totally isolating real time tasks.  
As far as we know, none of the existing works propose a soft real time CNN inference scheduler for GPU on the edge. We would like to ask this question: is it possible to provide soft real time  inference services to clients of an edge server while preserving high throughput?

 \begin{figure}
  	\centering
  	\includegraphics[width=1\linewidth]{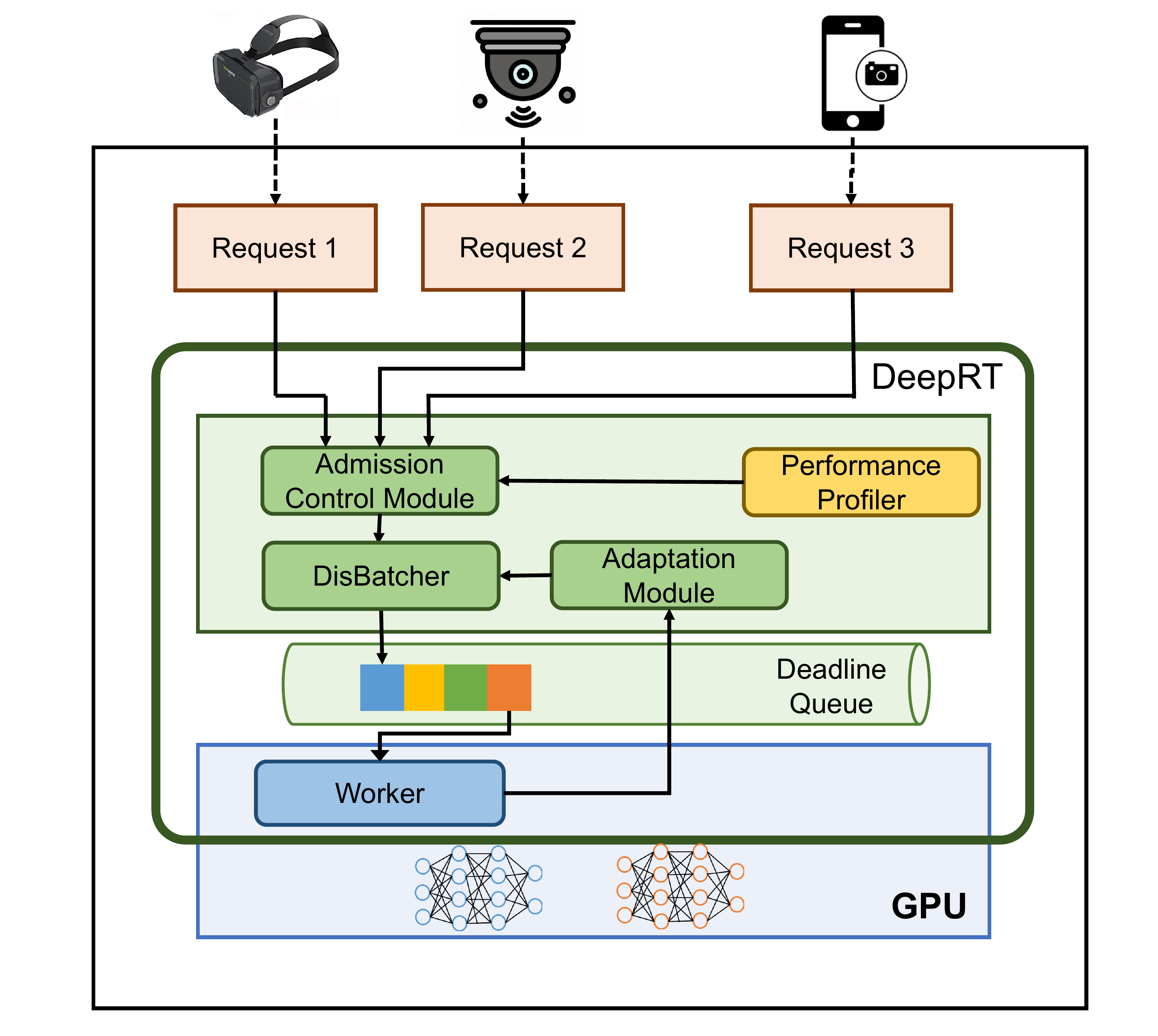}
	\caption{DeepRT system overview.}
	\label{fig:overview}
	\vspace{-5mm}
\end{figure}

Designing such a system is challenging. First, in order to guarantee soft real time services while maintaining throughput performance, it is necessary to inspect how different scheduling factors, including the number of concurrent models and batching, affect latency and throughput performance of the requests. Second, based upon the understanding of the complicated relationship between scheduling factors and system performance, and the fact that GPU operations are non-preemptive \cite{elliott2013gpusync}, how should we schedule the processing of requests to meet their deadlines while maintaining high throughput? 
Specifically, if we employ adaptive batching that batches data from different requests to increase throughput, we need to be aware that data from different requests arrive independently, and that data arriving early have to be queued to wait for  data from other requests in the same batch.
How do we determine batch size to meet deadlines of all requests? And when we obtain the batches, how to schedule their execution on GPU non-preemptively?
Third, after we determine the scheduling algorithm, we need to design an appropriate admission control test in order not to overload the system. If some tasks don't proceed as expected and cause deadline misses, we need a mechanism to tackle overruns and resolve the issue of deadline misses as fast as possible.

 In order to handle these challenges, we have done a comprehensive analysis of how inference performs under different factors, and propose a GPU inference scheduler, \name, which is able to provide soft real time inference services for multiple requests made from edge clients. Each request is to perform inference with a client-specified CNN model on a video consisting of a series of video frames which arrive at the system periodically.
 Specifically, we first use an edge  inference platform developed by NVIDIA, Triton Inference Server \cite{triton}, to study the latency and throughput performance characteristics when there are multiple CNN model instances to be executed and when request data are batched into different batch sizes. We have two key findings: (1) Executing multiple model instances concurrently on GPU doesn't significantly improve throughput. On the contrary, it greatly increases latency and under some circumstances makes it very difficult to estimate a worst-case latency. (2) Batching is able to increase system throughput  and  outperforms concurrent model instance execution in throughput increase, but it also sacrifices inference latency.
 In light of these observations, we design the {\name} system (see Figure \ref{fig:overview}), and within {\name} we put forward a batching mechanism, DisBatcher, which is able to batch as many data in one batch from different requests as possible, and we propose to execute the inference CNN models over these batched data sequentially instead of concurrently. The ordering of execution is determined by the  Earliest Deadline First (EDF) algorithm, since EDF is optimal in non-idling non-preemptive scheduling. 
 We also propose a two-phase Admission Control Module which determines whether new requests should be admitted, and an Adaptation Module which makes adaptation decisions in case of job overruns. We show that DisBatcher guarantees real time processing of all  requests admitted by the Admission Control Module. Our {\name} design also makes it easy to support non-real time requests, by batching them with the DisBatcher and assigning the batched data from non-real time requests with a low  priority. 

Overall, this paper makes the following contributions:

\begin{itemize}
    \item We perform a systematic analysis of latency and throughput performance of CNN inference under multitenancy situation.
    \item We propose a complete set of solutions of a soft real time CNN inference scheduler for GPU on the edge, including (1) a Performance Profiler, (2) a two-phase Admission Control Module , (3) a DisBatcher mechanism which batches image frames from multiple requests, (4) using EDF to schedule the batched jobs sequentially, and (5) an Adaptation Module. As far as we know, we are the first to design a soft real time CNN inference system over GPU on the edge.
    \item We  conduct comprehensive experiments to validate the performance of this system.
\end{itemize}

In this work we assume that each  request is a video. 
With minor modifications our system can also support the processing of other types of data on GPU, \emph{e.g.}, IoT sensory data. 
\ccomment{
\textcolor{red}{We focus on building a GPU scheduler, and leave the design of a scheduler involving data communication to future work.}}
It is also worth mentioning that  we target building a GPU scheduler for edge servers and the GPU uses the CUDA framework. This is the most common hardware setting in edge based inference systems. Support for other accelerators such as FPGA and TPU and for other frameworks such as OpenGL are left to future work. 

\ccomment{
The rest of the paper is organized as follows.  In Section \ref{sec:char} we study the characteristics of CNN inference under different scheduling factors. Afterwards, in Section \ref{sec:model}, we present how scheduling is performed in \name.  Section \ref{sec:design} describes {\name} system design. We discuss implementation details in Section \ref{sec:impl}.  Section \ref{sec:eval} presents our evaluation results. In Section \ref{sec:relatedwork}, we introduce state-of-the-art works on similar topics.  Section \ref{sec:conc} concludes the paper.
}
  
\vspace{-2mm}
\section{GPU Execution Characteristics}
\label{sec:char}

\begin{figure*}
  \centering
  \begin{subfigure}{.28\textwidth}
  \centering
  \vspace{-12pt}
  \includegraphics[width = 1\linewidth]{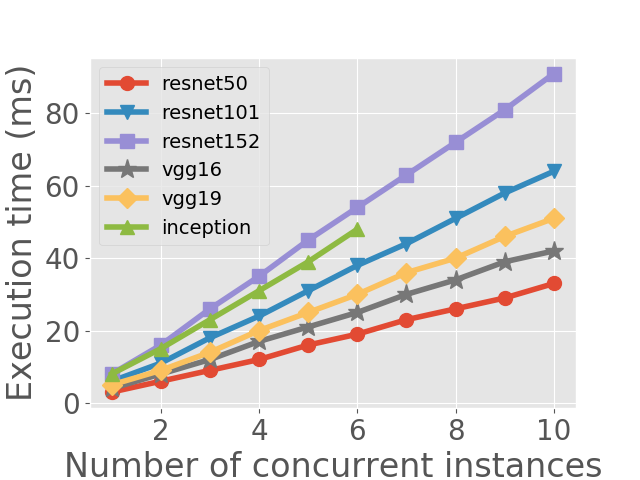}
  \caption{Median execution time when executing multiple instances of the same model.}
  \label{fig:same-latency}
  \end{subfigure}
  \hspace{15pt}
  \begin{subfigure}{.28\textwidth}
  \centering
  \vspace{-12pt}
  \includegraphics[width = 1\linewidth]{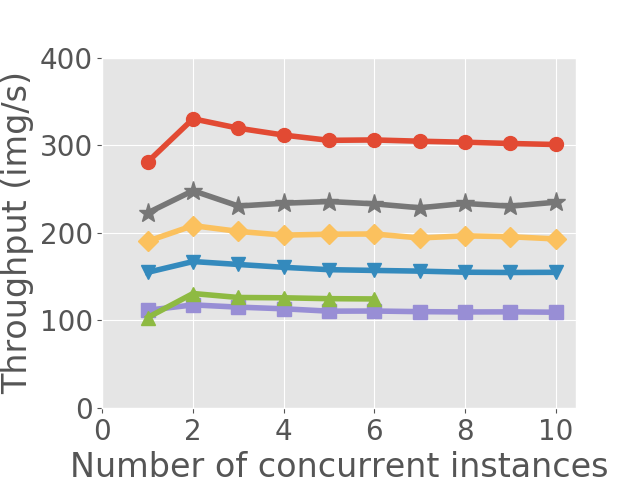}
  \caption{Overall throughput when executing multiple instances of the same model.}
  \label{fig:same-thru}
  \end{subfigure}
  \hspace{15pt}
  \begin{subfigure}{.28\textwidth}
  \centering
  \vspace{-12pt}
  \includegraphics[width = 1\linewidth]{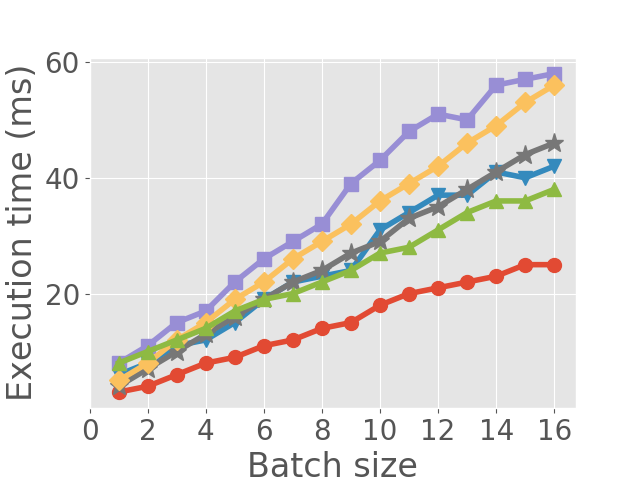}
  \caption{Median execution time when processing data in batches.}
  \label{fig:batch-latency}
  \end{subfigure}
    \begin{subfigure}{.28\textwidth}
    \centering
    \vspace{-3pt}
  \includegraphics[width = 1\linewidth]{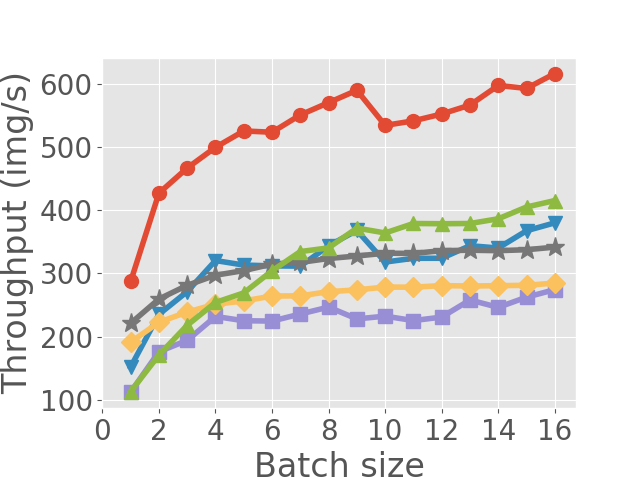}
  \caption{Overall throughput when processing data in batches.}
  \label{fig:batch-thru}
  \end{subfigure}
  \hspace{15pt}
   \begin{subfigure}{.28\textwidth}
   \centering
   \vspace{-3pt}
  \includegraphics[width = 1\linewidth]{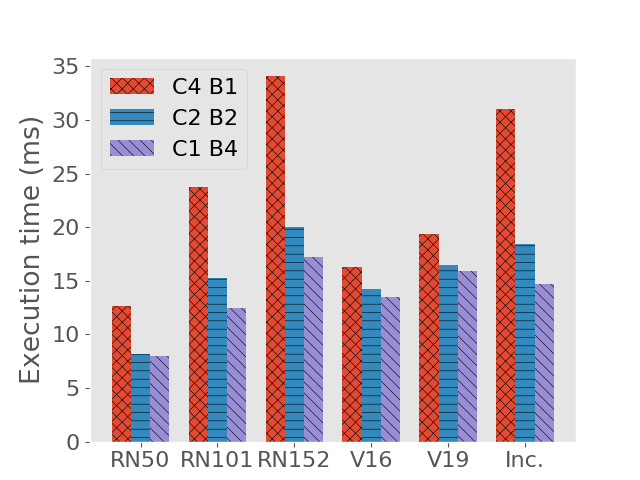}
  \caption{Comparison of execution time between concurrent execution and batch processing. }
  \label{fig:comp-latency}
  \end{subfigure}
  \hspace{15pt}
     \begin{subfigure}{.28\textwidth}
     \centering
     \vspace{-3pt}
  \includegraphics[width = 1\linewidth]{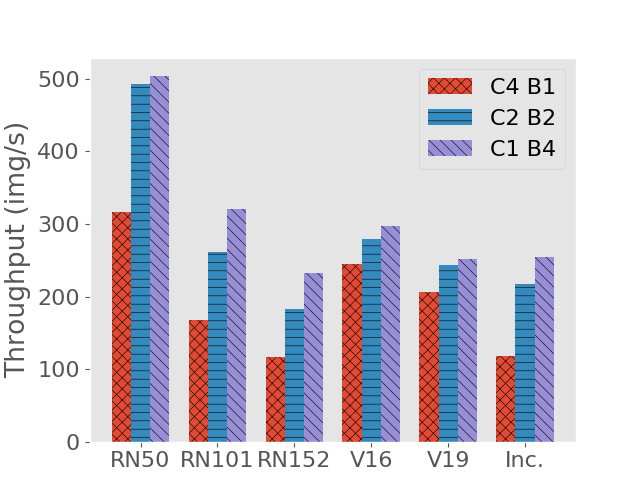}
  \caption{Comparison of throughput between concurrent execution and batch processing.}
  \label{fig:comp-thru}
  \end{subfigure}
  \vspace{2pt}
  \caption{These figures show the execution time and throughput performance under different concurrency and batch size conditions. In (e)  and (f), ``Cx By'' means running x model instances concurrently while each instance processing batches sized y. ``RN'' refers to ResNet, ``V'' refers to VGG, ``Inc.'' refers to Inception.}\label{fig:sec2}
  \vspace{-15pt}
\end{figure*}

While there has been a tremendous amount of research endeavor and industrial solutions aiming at  optimizing the processing latency and throughput of a single deep learning request on GPU, \emph{e.g.}, model pruning, fewer works focus on the handling of multiple concurrent requests, which is in fact an important  scenario in edge computing. In order to design an edge based system aiming at providing soft real time inference services for multiple clients, first we need to understand how GPU behaves under multiple requests. In this section, we show the performance characteristics of processing multiple deep learning requests on GPU, which is the foundation of our soft real time GPU scheduler. 

\subsection{Experimental Settings for the Analyses}

We mainly focus on the latency 
and throughput performance of CNN inference.  When batching is enabled, the  latency $l$ of performing CNN inference upon some input data consists of two parts:
\begin{equation}
\vspace{-2pt}
    l = l_q + l_e,
    \vspace{-2pt}
\label{eq:latency}
\end{equation}
where $l_e$ is the real execution time of performing CNN inference upon input data on GPU, and $l_q$ is the queuing time spent by some input data that arrive early in waiting for input data that arrive late. $l_q$ hinges on the specific design of the inference system, so throughput this section we only measure $l_e$ and call it \emph{execution time} to avoid ambiguity.

All the performance measurements are carried out using a mature cloud and edge inference solution, Triton Inference Server \cite{triton},  developed by NVIDIA. The hardware setting is introduced in Section \ref{sec:impl}. The way we measure execution time and throughput is as follows.
We use Triton's \texttt{perf\_analyzer} to send  requests to Triton server and record median execution time and throughput. 
Each time we send one or several requests to the server for inference, depending on the concurrency number, and each request may contain one image or a batch of multiple images; when  inference result is received from the server, we immediately send out another request(s). This process is repeated over a fixed time interval, $20$ seconds. 
We use images sized $3\times 224\times 224$ (RGB channels $\times$ height $\times$ width), which is the default image size in Triton. We choose $6$ widely used deep learning models -- ResNet50, ResNet101, ResNet152, VGG16, VGG19, and Inception-v3. They belong to $3$ types, ResNet \cite{he2016deep}, VGGNet \cite{simonyan2014very}, and Inception-v3 \cite{szegedy2016rethinking}. In our setting, ResNet and VGGNet models are built upon the ONNX framework \cite{bai2019}, and Inception-v3 is of the GraphDef format \cite{tensorflow}, which is a tool used by Tensorflow to represent models. Our setting covers different types of models, different model sizes in each type, and different frameworks to show the universality of our conclusions.

\vspace{-5pt}
\subsection{Concurrent Execution of Models}

We first study the performance characteristics when executing multiple models concurrently on GPU. To be more specific, there are multiple model instances loaded on GPU, and these models process image frames received from clients all at the same time; we study how concurrency affects the execution of each model. This analysis can be further divided into two parts: concurrent execution of multiple instances of the same model on GPU, and concurrent execution of different models.

\textbf{Concurrent execution of the same model.} When different clients request to process their data with the same model, except for batching the data and processing them as a whole in one model instance, another common approach taken by Triton and some  platforms is to replicate the same model on GPU to get several model instances, and to process each request with one of the model instances concurrently. In this part we study the performance of running different number of duplicate instances of the same model concurrently (see Figure \ref{fig:same-latency} and  \ref{fig:same-thru}).
For Inception we show its performance  up to a concurrency number of $6$, as more concurrent instances overload the system. 
 We have also recorded the variance of different executions, but  the variances are too small to display.

\ccomment{
\begin{figure}
  	\centering
  	\includegraphics[width=1\linewidth]{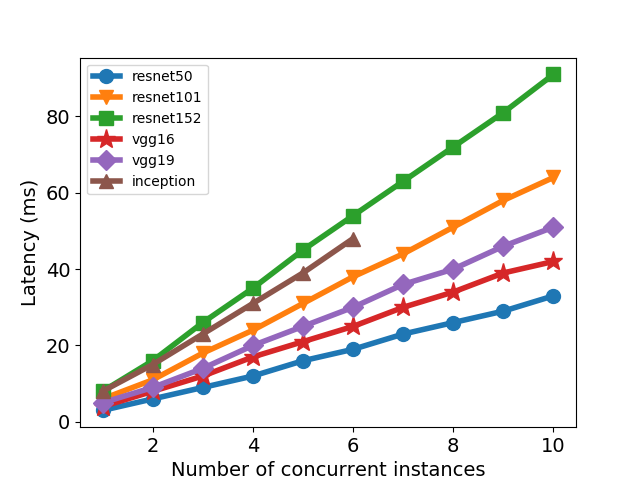}
	\vspace{1pt}
	\caption{.}
	\label{fig:same-latency}
	\vspace{-5mm}
\end{figure}
\begin{figure}
  	\centering
  	\includegraphics[width=1\linewidth]{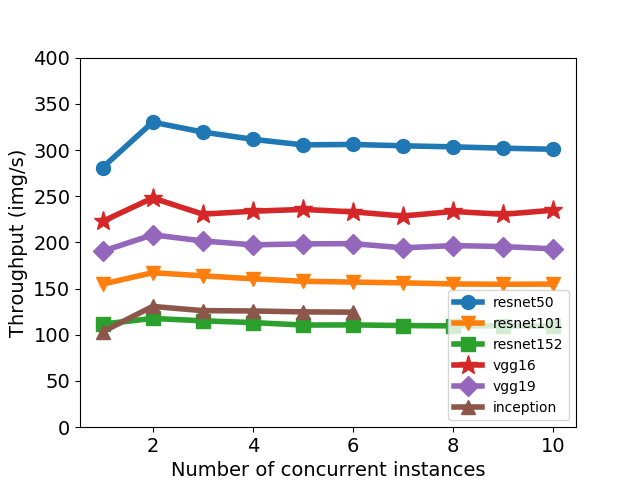}
	\vspace{1pt}
	\caption{.}
	\label{fig:same-thru}
	\vspace{-5mm}
\end{figure}}


 From Figure \ref{fig:same-latency} we can see the linear relationship between execution time and the number of concurrent instances. 
 We think this is due to how CUDA schedules multiple programs. When multiple programs run concurrently on GPU, CUDA schedules their warps (a warp contains multiple CUDA threads which can be executed in parallel) with a time sliced scheduler \cite{wang2017quality}.  Warps from different contexts cannot be executed simultaneously and some warps have to wait in queue for their time slice share. When we execute multiple instances of the same model, since all instances have exactly the same warps but the warps are in different contexts, the execution time of each warp 
 grows linearly with the concurrency number.
 Figure \ref{fig:same-thru} shows how throughput changes with the concurrency number. We can see that if we increase concurrency number to $2$, there is a slight increase in throughput, but afterwards the throughput stays at a stable value.

 A crucial conclusion we draw in this part is that, although increasing concurrency to a certain level  can improve inference throughput by a small margin, execution time increases linearly as concurrency number grows. Let's look at an example to see how this observation affects scheduling algorithm design. Imagine there are two images to be processed with a specified model, say, ResNet50. Processing the two images one by one inside one single model instance instead of processing them concurrently in two model instances,  achieves a slightly lower throughput but the average latency is reduced by $25\%$, since the first image is finished early and its execution time is not affected by the concurrent execution of the second image. Therefore, the latency of the first image is reduced by half compared to running concurrently, while  the latency of the second image remains approximately the same as when running concurrently.
 
\begin{center}
\begin{table}
\footnotesize
\centering
\begin{tabular}{ c| C{16pt} C{18pt} C{19pt} C{19pt} C{18pt} C{18pt} C{18pt} }
\hline
\multicolumn{7}{c}{Execution time (ms)} \\
\hline
   & - & RN50 & RN101 & RN152 & V16 & V19 & Inc. \\ 
 \hline
 RN50 & 3.5 {\scriptsize(0.0)} & 6.4 {\scriptsize(0.0)} & 6.9 {\scriptsize(0.6)} & 7.1 {\scriptsize(0.4)} & 11.2 {\scriptsize(0.0)} & 11.6 {\scriptsize(0.6)} & 5.8 {\scriptsize(0.8)} \\  
 \hline
 RN101 & 6.4 {\scriptsize(0.0)} & 11.1 {\scriptsize(0.4)} & 11.8 {\scriptsize(0.0)} & 12.0 {\scriptsize(0.5)} & 18.0 {\scriptsize(0.1)} & 20.3 {\scriptsize(0.1)} & 9.0 {\scriptsize(1.3)} \\
 \hline
 RN152 & 9.0 {\scriptsize(0.1)} & 15.5 {\scriptsize(0.3)} & 16.4 {\scriptsize(0.5)} & 16.8 {\scriptsize(0.1)} & 24.6 {\scriptsize(0.1)} & 27.6 {\scriptsize(0.1)} & 14.6 {\scriptsize(0.4)}\\
 \hline
 V16 & 4.5 {\scriptsize(0.0)} & 5.9 {\scriptsize(0.2)} &	6.0 {\scriptsize(0.4)} &	6.3 {\scriptsize(0.2)} &	8.1 {\scriptsize(0.0)} &	8.8 {\scriptsize(0.0)} &	5.2 {\scriptsize(0.5)}\\
 \hline
 V19 & 5.3 {\scriptsize(0.1)} & 6.9 {\scriptsize(0.4)} &	7.0 {\scriptsize(0.4)} &	7.4 {\scriptsize(0.2)} &	8.8 {\scriptsize(0.2)} &	9.6 {\scriptsize(0.0)} &	6.1 {\scriptsize(0.5)}\\
 \hline
 Inc & 9.3 {\scriptsize(1.4)} & 25.3 {\scriptsize(0.6)} & 29.0 {\scriptsize(0.5)} &	28.9 {\scriptsize(0.4)} &	37.6 {\scriptsize(1.3)} &	42.9 {\scriptsize(0.8)} &	15.2 {\scriptsize(0.5)}\\
 \hline

\end{tabular}

\vspace{6pt}

\begin{tabular}{ c|C{18pt} C{20pt} C{24pt} C{24pt} C{18pt} C{18pt} C{18pt} }
\hline
\multicolumn{7}{c}{Throughput (img/s)} \\
\hline
   & - & RN50 & RN101 & RN152 & V16 & V19 & Inc. \\ 
 \hline
 RN50 & 282.1 & 155.8 &	145.4 &	143.4 &	89.2 &	87.1 &	174.7  \\ 
 \hline
 RN101 & 154.9 & 90.4 &	84.7 &	82.1 &	55.6 &	51.0 &	103.1 \\
 \hline
 RN152 & 111.6 & 64.2 &	60.9 &	59.6 &	40.7 &	36.8 &	69.1\\
 \hline
 V16 & 222.8 & 178.0 &	166.8 &	162.7 &	123.4 &	113.1 &	187.4\\
 \hline
 V19 & 190.5 & 148.3 &	144.4 &	142.9 &	113.1 &	103.8 &	162.1\\
 \hline
 Inc & 105.9 & 39.7 &	34.4 &	34.5 &	26.6 &	23.6 &	65.5\\
 \hline

\end{tabular}
\vspace{5pt}
\caption{Execution time and throughput when running two different model instances concurrently.}
\label{tb:sec2-diff}
\vspace{-15pt}
\end{table}
\end{center}

\vspace{-15pt}
\textbf{Concurrent execution of different models.} A common practice of Triton Inference server and other solutions when different clients send requests to process data with different models is to execute the models simultaneously. This part analyzes the performance characteristics when different models run concurrently on GPU. Among the $6$ models introduced previously, in each experimental run we choose two models and execute them concurrently\footnote{In this analysis, we study how the execution of a model is affected by another model instance. We don't analyze the situations where there are three or more instances since the current setting already shows the complicated interference between model executions.}. The other experimental setups are the same as in the previous part.

In Table \ref{tb:sec2-diff}, 
we show the  median execution time and throughput when executing a model specified by the leftmost column concurrently with another model specified by the uppermost row. As a comparison, we use the columns marked by ``-'' to  show the performance when a model is executed alone. The data show that when a model M is executed along with different other models, its execution time and throughput performance vary greatly\footnote{There is an observable trend, however, that when a model is executed concurrently with models of the same type (\emph{e.g.}, a ResNet101 model and a ResNet152 model), the performance tend to be similar, but performance discrepancy still exists. This trend is a piece of evidence for our following hypothesis. The kernels of models belonging to the same type have similar sizes, thus they cause similar amount of interference.}.  Explaining such an observation  requires a detailed study of how internal scheduling happens at a low level on GPU, which is very difficult since GPU drivers are not open sourced. Our hypothesis  is that different kernel (a CUDA function) sizes of different models cause different slowdowns. CUDA uses a time sliced scheduler to schedule programs from different contexts (corresponding to different models) and the kernels are the smallest scheduling unit. They get scheduled non-preemptively on GPU. Different models are composed of distinct numbers and types of kernels; when a model M is executed concurrently with model N and the kernels of M are larger in sizes but smaller in quantity, model M will have more GPU time share and thus exhibit higher throughput and smaller execution time.

\ccomment{
\begin{figure}
  	\centering
  	\includegraphics[width=1\linewidth]{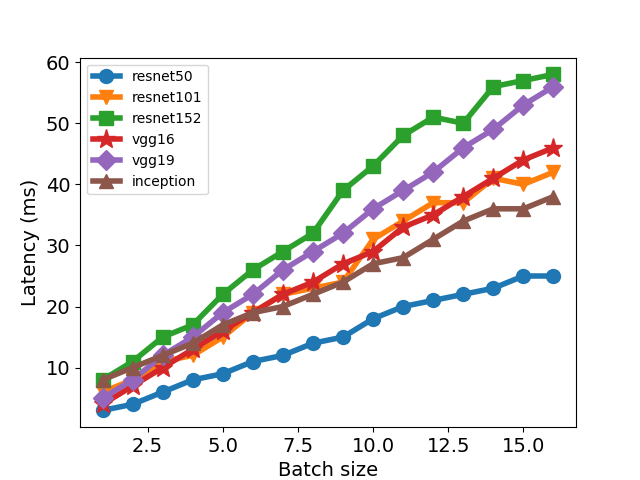}
	\vspace{1pt}
	\caption{.}
	\label{fig:batch-latency}
	\vspace{-5mm}
\end{figure}
\begin{figure}
  	\centering
  	\includegraphics[width=1\linewidth]{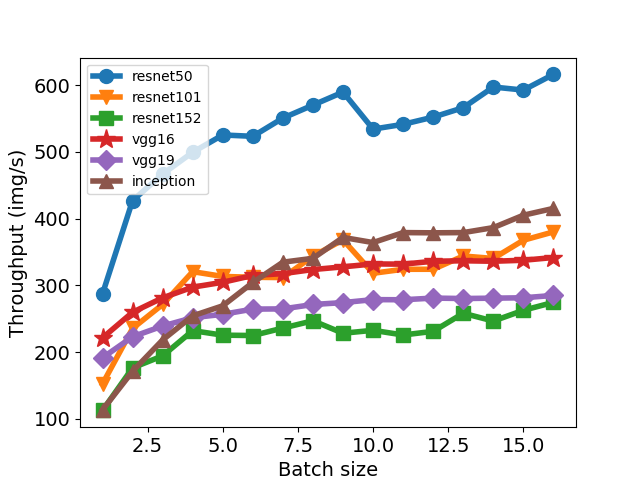}
	\vspace{1pt}
	\caption{.}
	\label{fig:batch-thru}
	\vspace{-5mm}
\end{figure}}

The observation we make from this part, is that concurrently executing multiple instances of different models results in complex interference between models. If we want to design a real time system, we need to figure out the interference in any subset of all admitted models in order to obtain their worst-case execution times. When there are $z$ models, there are  ${z\choose 2}$ combinations of models to profile if we only execute two model instances concurrently. In fact, if we allow any number of concurrent model instances, we need to profile $\sum_{k=1}^z {z\choose k} = 2^z$ combinations of models, which is practically impossible. However, if we manage to execute the different model instances sequentially instead of concurrently, there won't be any interference between the sequential model executions.

\vspace{-5pt}
\subsection{Inference in Batches.}

A well-known conclusion about GPU-based deep learning model training and inference is that batch processing can boost their throughput. In this part we show the performance characteristics when we do inference in batches. The experimental setup of this part is the same as in the previous two parts, except that each time we only execute one model instance  instead of concurrently executing two or more instances. The comparison between different batch sizes is shown in Figure \ref{fig:batch-latency} and Figure \ref{fig:batch-thru}. As expected, batching increases throughput with a cost of higher execution time. 

\vspace{-5pt}
\subsection{Concurrent Execution vs Batch Processing}

\ccomment{
\begin{figure}
  	\centering
  	\includegraphics[width=1\linewidth]{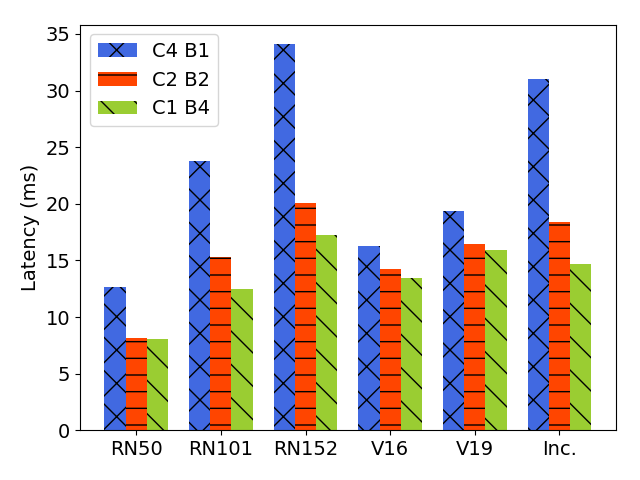}
	\vspace{1pt}
	\caption{.}
	\label{fig:comp-latency}
	\vspace{-5mm}
\end{figure}
\begin{figure}
  	\centering
  	\includegraphics[width=1\linewidth]{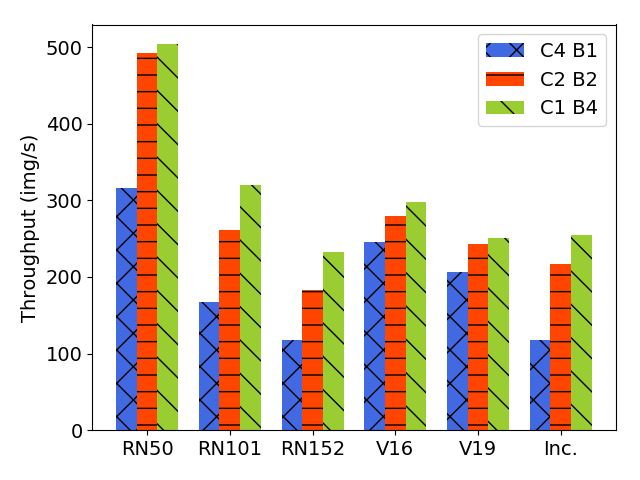}
	\vspace{1pt}
	\caption{.}
	\label{fig:comp-thru}
	\vspace{-5mm}
\end{figure}}

Since concurrent execution of multiple instances and batching inputs into larger tensors can both process multiple requests simultaneously, in this part we compare how these two approaches affect system performance. On each of the aforementioned $6$ models, we fix the number of concurrently processed requests (concurrent model instances $\times$ batch size) and  see how the execution time and throughput vary. The result is shown in Figure \ref{fig:comp-latency} and Figure \ref{fig:comp-thru}, where we compare concurrency - batch combinations of $4\times 1$, $2\times 2$, and $1\times 4$. 
We can see that batch processing
should be favored over concurrent model instance execution in terms of increasing system throughput and reducing request execution time.

\subsection{Summary of Observations}

In summary, We make three  observations: 

\begin{itemize}
    \item Executing two or more models simultaneously does not notably improve  system throughput,  and  it   increases latency of each request due to increased execution time. 
    \item When executing  multiple instances from different models on the same GPU, there exists interference between concurrent executions, 
    making it cumbersome to profile or estimate the  worst-case execution time, hindering the design of a real time system.
    \item Processing input data in batches increases system throughput, much more than concurrent model execution does, with a sacrifice of increased request execution time.
\end{itemize}

Two key takeaways from these observations are: (1) Concurrently executing multiple model instances does not provide high value to our goal of guaranteeing maximum latency  while maintaining high throughput; instead, the interference it introduces makes designing a real time system difficult. So we would like to execute model instances sequentially instead of concurrently to avoid interference. (2) Batching  increases throughput
, but we need to make sure that the increased latency   (increased execution time plus queuing time when an image frame waits for other frames belonging to the same batch)  
does not cause deadline misses.
\vspace{-2mm}
\section{{{\name}}   Scheduling Scheme}
\label{sec:model}

\begin{figure*}
  	\centering
  	\includegraphics[width=1\linewidth]{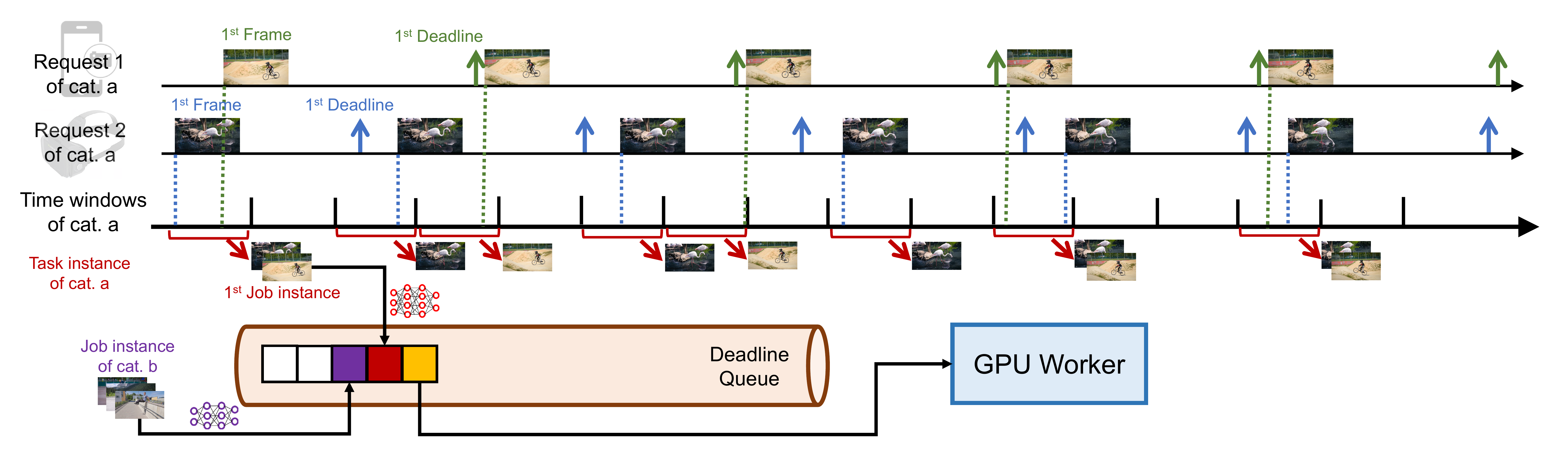}
	\vspace{-15pt}
	\caption{An illustration of our batching mechanism. 
We show how two requests of category \emph{a} are batched into a task instance. The dashed lines represent the arrival times of  frames, and the arrows pointing upward represent their deadlines. Frames of category \emph{a} are batched into a job instance if they arrive during the same time window. Then the  job instance are pushed to a deadline queue for processing.}
	\label{fig:tw}
	\vspace{-5mm}
\end{figure*}

Based upon the observations we make in the previous section,  we propose the soft real time scheduling scheme of {\name}, which is proven to meet the deadlines of  admitted requests while being able to process input data in batches in order to boost inference throughput. Specifically, we propose a batching approach called DisBatcher, which  batches input data from multiple requests according to their release times and deadlines, and pushes the batched data to a deadline queue to be processed by the GPU.  We propose to schedule the execution of the batches with EDF. Then we prove that DisBatcher ensures all deadlines of the input data are met if the batches can be scheduled with EDF.

\subsection{System Model}

\textbf{Data Model.} This paper assumes that the GPU resources  on the edge are shared by multiple users. Each user sends a request to {\name}, and each request corresponds to a video, the \emph{frames} of which need to be processed by a user-specified CNN model. 
Different users may request different CNN models.
When a request comes, if {\name} confirms that the new request and all the existing requests in the system can  be scheduled, {\name} will admit the new request\footnote{In this paper, we use admitted requests and requests interchangeably, and we call requests that wait to be verified pending requests.}.
 The video frames of an admitted request arrive at {\name} in an online manner, and the interval between two frames is determined by a frame rate, or \emph{period}\footnote{We only consider the scenario where all frames are raw,  uncompressed video frames. Decompression of compressed frames should be coordinated by a CPU scheduler and is not the scope of {\name}.  Also, we assume that the frames in a video arrive one by one, but {\name} can also handle the situation where a chunk of frames arrive at a time, used by video coding techniques like H.264. }. Each request also has a \emph{relative deadline}, indicating the desired maximum latency of performing inference on each of the video frames. The relative deadline of a request is not necessarily equal to its period. Different videos may have different input  \emph{shapes}, which equal to the number of channels (most videos use RGB channels so this number is $3$) $\times$ frame height $\times$ frame width.





\textbf{Execution Model.} The two key takeaways   from Section \ref{sec:char} drive us to avoid concurrent model execution and leverage batching to increase system throughput. We propose the following execution model for {\name}. When a frame from a request arrives, {\name} does not process this  frame immediately. Instead, {\name} waits for  frames of the same shape requiring the same CNN model, batches these  frames together, and processes the batched data on GPU. 
In fact, all  frames that are of the same shape and require the same model can be batched together regardless of which request they belong to, while  frames with different shape or require different models should not be batched because GPU cannot execute in parallel the same kernels for these frames. Since the batch of frames from multiple requests is what the GPU actually executes, we call the job of processing a batch of frames a \emph{job instance}. 
{\name}  processes multiple job instances  \emph{one at a time} instead of concurrently. In {\name}, the frames with the same shape requiring the same model are said to be of the same \emph{category}. Likewise, we call the requests containing the same category of frames to be of the same category.  Executing job instances sequentially instead of concurrently makes Equation \ref{eq:latency} become $l=l_{q_b}+l_{q_j}+l_e$, where $l_{q_b}$ denotes frame queuing time waiting for other frames in the same batch, and $l_{q_j}$ denotes queuing time of a job instances. In order to provide a guarantee on $l$, two questions need to be answered: (a) How many frames should be in each batched job instance, or how does {\name} determine which  frames should be put into the same batch? (b) GPU operations are non-preemptive, in the sense that operations already launched on GPU cannot be preempted. How to schedule the non-preemptive batched job instances?

\vspace{-10pt}
\subsection{{\name} Batching Approach}

This subsection answers the first question. There are already research efforts \cite{crankshaw2017clipper}\cite{fang2019serving} or industrial solutions \cite{triton} which process batches containing fixed or adaptive number of images from multiple requests concurrently. While these approaches are capable of increasing inference throughput and reducing latency, their approaches don't support real time services, since they don't have a deadline-centric soft real time design. A real time system with batching enabled has to make sure the queuing time of a frame arriving early waiting for frames arriving late does not cause a deadline miss of the early frame.

In order to guarantee real time processing of  frames, instead of answering the question ``what is the right number of  frames inside a batch'',  we try to answer another question ``when should a frame be put inside a batch''. We propose  the DisBatcher approach to answer this question. DisBatcher divides time into contiguous time intervals with the same length, called \emph{time windows}. The end time of a window coincides with the start time of next time window.  Frames that fall into the same time window should be batched together into one job instance regardless of which request the frames belong to, as long as the frames are of the same category, and the batching action happens at the end of a time window. DisBatcher sets the relative deadline of a job instance to be the length of one time window, which means that a job instance generated at the end of one time window, should be completed before the end of the next time window. For simplicity, We call the end time of a time window, which is also the start time of the next time window,  a time window \emph{joint}. As we previously mentioned,  frames of different categories should not be batched together. DisBatcher creates different and independent time windows for each category of frames. 
With this time window approach, we transform the processing of one category of requests into a series of non-preemptive job instances, or a \emph{task instance}. A task instance's relative deadline is equal to its period, which is the length of the  time window used to generate the task instance.
This task instance is not strictly a periodic task, since each job instance has different execution times due to different number of  frames inside the batches, and some job instances have $0$ as their execution times (no frames in this batch).   
An example of using time windows to batch  frames is shown in Figure \ref{fig:tw}.

How to  set appropriate time window lengths in order to ensure that  for all frames their latencies do not exceed their deadlines remains a challenge. DisBatcher utilizes a deadline-centric design.
For one category of  requests, DisBatcher sets the time window length to be half of the smallest relative deadline among all requests, regardless of other parameters such as request arrival periods. In fact, we have the following theorem.

\vspace{-5pt}
\theoremstyle{definition}
\begin{theorem}
Given a set of deep learning inference requests $I$. Each request $I^g_m, g\in \Gamma, m\in M_g$ consists of a series of  frames which need to be processed by GPU, where $\Gamma$ denotes the set of all request categories and $M_g$ denotes the set of all requests in category $g$. 
Each request also has a relative deadline $d^g_m$. A request is schedulable if the latency of each frame is smaller than or equal to $d^g_m$. Likewise, we call a job instance schedulable if its latency is no larger than its deadline.

If we use the time window scheme to batch process the  frames with the time window length of each category $W_g$ set to
\vspace{-2pt}
\begin{equation}
\vspace{-2pt}
    W_g = \frac{1}{2}\min_{m\in M_g} d^g_m, g\in \Gamma,
\end{equation}
all requests in $I$ are schedulable if the periodic batched job instances are schedulable.
\end{theorem}
\begin{proof}
 When the length of time windows of a request category is equal to a half of the smallest relative deadline in this request category, for all the  frames in all the requests of this  category, there are at least two time window joints between the arrival time and the deadline of a frame (inclusive). 
 This argument is illustrated in Figure \ref{fig:tw}, where the length of time windows is set to be smaller than half of the  smallest relative deadline.

All  frames of the same category that arrive at the same time window are batched at the end of the current time window , which is also the first time window joint they encounter. Since the batched job instance takes the next time window joint as its deadline, all  frames can be finished processing before the second time window joint if the batched job instance meets its deadline. Since there are at least two time window joints between the arrival time of a frame and the deadline of a frame, the deadline of a job instance will be earlier than the deadline of all its corresponding frames. Therefore, if the job instances are schedulable, so will be the frames.
\end{proof}

\vspace{-10pt}
\subsection{Job Instance Scheduling with EDF}

With the DisBatcher approach, we transform the problem of how to batch and schedule frames in real time  to the problem of scheduling a set of non-preemptive periodic task instances. The new task instances are not strictly the traditional non-preemptive periodic tasks, since job instances of the same task instance have variant execution times. 

The non-preemptive periodic tasks whose jobs have variant execution times are called non-preemptive multiframe tasks \cite{moyo2010schedulability}\cite{chakraborty2002schedulability}\cite{baruah2010preemptive}. Non-preemptive periodic tasks with fixed execution times are a special case of non-preemptive multiframe tasks. There are two types of scheduling algorithms for non-preemptive multiframe tasks --- algorithms that do not permit the processor to be idle when there are  jobs that have been released but have not completed execution  (non-idling), and algorithms that allow idle times (idling). For {\name}, we choose non-idling Earliest Deadline First (EDF) as its scheduling algorithm, the reasons of which are twofold. 

First,  the inserted idle times lead to a waste of precious GPU computation power, sacrificing the total throughput of the system. Second, although some idling algorithms can perform better than non-idling algorithms in terms of the number of schedulable tasks \cite{ekelin2006clairvoyant}, they often depend on complicated heuristics, and their performance gains  are only demonstrated through simulations. In fact, finding an optimal schedule in idling non-preemptive context is an NP-hard problem \cite{gary1979computers}. On the contrary, EDF has been shown to be an optimal scheduling algorithm for  non-preemptive multiframe tasks in non-idling context \cite{george1995optimality}\cite{baruah2006schedulability}\cite{chakraborty2002schedulability}.
  In Section \ref{sec:design}, we will discuss how {\name} performs admission control in order to make sure that all job instances are schedulable.

It is worth mentioning that the scheduling scheme of {\name} makes it easy to support non-real time requests. We treat non-real time requests in a similar way with real time requests. We also use DisBatcher to batch non-real time requests to transform the frames into  task instances. However, we don't batch non-real time requests together with real time requests for performance isolation;  we use a large time window for non-real time requests to  make sure they obtain a low deadline priority, and we impose them a large arrival period regardless of their true arrival periods so that they don't aggregate to large batches and cause too much priority inversion.

\vspace{-2mm}
\section{{{{\name}}} System Design}
\label{sec:design}

In this section we present the whole system of {\name}, shown in Figure \ref{fig:overview}. {\name} is a scheduling system built on top of the scheduling scheme presented in the previous section, aiming at providing a soft real time inference service for CNN models on GPU. {\name} consists of $5$ parts: a Performance Profiler, a two-phase Admission Control Module, a DisBatcher, an Execution Worker, and an Adaptation Module.

\ccomment{
First, {\name} has a Performance Profiler which performs offline performance analysis to obtain the execution times of batched job instances of different batch sizes and different categories. 
When a new request arrives at {\name}, {\name} first routes this pending request to a two-phase Admission Control Module where {\name} performs an admission control test to determine whether the pending request is  schedulable under current system workload, based on the data provided by the Performance Profiler. If the new request together with current workload is schedulable, {\name} begins the processing of this request by passing its frames to a DisBatcher, where the frames get   batched with frames from other requests into job instances and queued. Then, an EDF scheduler commands a Worker on GPU to start processing the job instances. We also have an Adaptation Module which monitors the performance of the Worker and adjusts the execution plan when necessary.
}

\vspace{-5pt}
\subsection{Performance Profiler}
\label{ssec:perf_prof}


Our Performance Profiler works offline. For each deep learning model that we want to execute at the edge server, for each frame shape that {\name} permits as legitimate shape, and for  different batch sizes, the Performance Profiler executes each  batch of frames on GPU multiple times, and records the execution time of each run. For each setting, we obtain a list of running times  and take the worst-case running time\footnote{In practice we take 99 percentile running time to filter out outliers.}.
In this way, we create a lookup table containing the execution times of different sized batches of frames with different shapes processed by different deep learning models. 
Whenever a new request comes to the system, we look up this table, find the corresponding  model and  shape, and feed the results to Admission Control Module to make admission decisions.

\vspace{-5pt}
\subsection{Admission Control Module}

When a new request arrives at {\name}, it is first routed to the Admission Control Module. Since we target  building a soft real time system, {\name} is selective with the requests in case  too many requests  cause serious  deadline misses. The Admission Control Module decides whether a pending request is admitted to {\name}.
As we discussed in Section \ref{sec:model}, {\name} uses DisBatcher to transform frames into task instances, which are intrinstically non-preemptive multiframe tasks.
Therefore, performing admission control for DisBatcher based {\name} is equivalent to performing admission control for non-preemptive multiframe tasks.

Some past works propose to perform admission control  for EDF under the non-preemptive  multiframe workload scenario using  demand-bound functions \cite{baruah2010non}\cite{baruah2010preemptive}\cite{baruah2006schedulability}\cite{chakraborty2002schedulability}. 
A demand bound function represents the
maximum execution demand of a task set in any time interval of a given length. 
Then this approach compares the demand bound with available resources to decide whether a task set is schedulable.
This approach suffers from pseudo-polynomial complexity and inaccuracy of their approximate algorithms in calculating the demand bound functions. Another approach performs simulation based feasibility analysis \cite{moyo2010schedulability}. Basically this approach represents time  with a clock variable.  When the clock reaches the arrival time of a job, the job is released to a deadline queue.
It simulates the execution of the job by simply incrementing the clock by the job's worst-case execution time. Then it compares the virtual completion time of the job, which is the current value of the clock variable, with the job's deadline, to determine whether there is a deadline miss. Since different tasks may have different initial release times, this approach uses a tree to represent all possible execution sequences, making its time complexity non polynomial. 

Since the goal of {\name} is to provide real time inference while maintaining high throughput, the Admission Control Module should admit as many requests as possible but not too many requests which overload the system. Therefore, it has to perform an exact analysis of schedulability.   {\name} adopts the simulation based approach as  it is an exact analysis. 
The time complexity of this approach can be  greatly reduced in {\name} to linear  with respect to the number of frames. The reason is that in {\name} the requests for  inference all have specific release times, because 
users' video frames occur at specific times instead of at arbitrary times, and the release times are communicated with the Admission Control Module. Since we also know the time when every time window starts, we can  know the release times of all job instances. In this way, we know exactly when and in what order each job instance ``arrives'' at the GPU, so we can build an exact execution schedule in linear time instead of building a tree of execution sequences. Note that in building such an execution schedule, {\name} requires the execution times of different job instances obtained in Section \ref{ssec:perf_prof}. The assumption for the exact schedulability analysis is thus accurate job instance execution time profiling.

To further improve the Admission Control Module by reducing its complexity, before we run the simulation based schedulability test, we first use a utilization based test to filter out obviously infeasible requests. The goal of the utilization based Phase 1 test  is to  reject a pending  request as fast as possible which will obviously cause deadline misses if accepted. The simulation based Phase 2 test is an exact test, which refines the results from Phase 1 test and ultimately decides whether the new request gets admitted. 

\textbf{Phase 1.} In Phase 1, {\name} uses the utilization of task instances to reject a pending request that  will obviously cause deadline violations. We calculate and evaluate the utilization of task instances because task instances are actually executed by the GPU .
We define the average utilization of a task instance $s$ to be 
\vspace{-5pt}
\begin{equation}
    U_s = \frac{\sum_{i=1}^{N_s}E_i}{N_sP_s},
\label{eq:ut}
\vspace{-3pt}
\end{equation}
where $N_s$ denotes the total number of job instances in task instance $s$, and $P_s$ is the period of $s$ which equals to the time window length used to generate $s$. Within each period there is at most one job instance. $E_i$ denotes the execution time of a job instance $i$ belonging to $s$. We further define the average utilization of a task instance set $\Sigma$ to be $U = \sum_{s\in \Sigma}U_s$. 
\ccomment{
\begin{equation}
    U = \sum_{s\in \Sigma}U_s ~.
\label{eq:u}
\end{equation}
} The task instance set comprises task instances generated from all existing requests and the new pending request. 
Naturally in order not to overload the system,  we should  have $U \leq 1$, or $\sum_{s\in \Sigma}\frac{\sum_{i=1}^{N_s}E_i}{N_sP_s} \leq 1$. The complexity of calculating $\sum_{i=1}^{N_s}E_i$ is linear with respect to $N_s$, since we need to estimate the number of  frames that fall inside each time window, and look up the profiling execution time table to decide $E_i$. 

In order to boost the speed of Phase 1, we make an approximation when calculating $U_s$. As the complexity of calculating $\sum_{i=1}^{N_s}E_i$ comes from the fact that we don't know how many  frames fall into each batch, we use an average number of  frames over all time windows to estimate the exact number of  frames in each time window. We denote the set of requests of the same category with $I^g$. The average number of  frames inside a time window of a request $I^g_m\in I^g$ is $\frac{W_g}{p^g_m}$, where $W_g$ is the period of the corresponding time window, and  $p^g_m$ is the period of  request $I^g_m$. Therefore, for all the requests of category $g$, the average number of  frames in one time window, denoted by $n_g$, can be represented by 
$
    n_g = \lfloor\sum_{I^g_m\in I^g}\frac{W_g}{p^g_m}\rfloor.
$ 
We look up the execution time $E^{n_g}$of a batch of $n_g$  frames from the table and we can get an estimation of $U_s$:
\vspace{-5pt}
\begin{equation}
\vspace{-5pt}
    \Tilde{U_s} = \frac{E^{n_g}}{P_s}.
    \vspace{-1pt}
\end{equation}

Please note that Phase 1  accepts more pending requests than {\name} can handle while not rejecting feasible requests by underestimating total workload. The reasons of this underestimation are twofold. First, it's  clear that the sufficient necessary condition for a set of periodic preemptive tasks to be schedulable is that the  total utilization is no larger than 1; but this is only a necessary condition for preemptive multiframe tasks \cite{mok1997multiframe}, let alone   non-preemptive multiframe tasks in our scenario. Second, we use an estimated average utilization of a task instance which doesn't consider  peak utilization. Also, we use a floor operator in calculating $n_g$ above. 
 In this way, Phase 1 admission control can give prompt responses to some clients and reduces some workload for Phase 2.

\textbf{Phase 2.} While {\name}'s Phase 1 test admits requests generously, the Phase 2 test does an exact schedulability analysis to control admission to the system. The Phase 2 test consists of three sequential steps --- system state recording, pseudo job instances generation, and an EDF imitator algorithm.

 \begin{algorithm}[t]
\SetAlgoLined
\KwIn{Sorted deadline queue $Q$, list of job instances $L$ ordered by release times}
\KwOut{Whether the jobs are schedulable with EDF}
 Initialization: $t \gets 0$\;
 
 \While{$Q$ not empty {\normalfont \textbf{or}} $L$ not empty}{
 \eIf{$Q$ is empty}{
   release job $i$ from $L$ to $Q$\;
   $t\gets R_i$\;
   }{pop job $k$ from $Q$\;
     $t\gets t+E_k$\;
     \If{$t>R_k+D_k$}{\Return not schedulable}
     \While{$L$ not empty {\normalfont \textbf{and}} $R_{L[1]}<t$}{
     release job $i$ from $L$ to $Q$\;}
   }}
   \Return schedulable
 
\label{al:edfimi}
 \caption{EDF imitator algorithm}
\end{algorithm}

In the first step, system state recording, the Admission Control Module captures the current system state, which includes four parts: (1) the number of frames of each category that have already arrived at {\name} and  wait to be batched by the DisBatcher,  (2) the already batched  job instances  in the  deadline queue waiting to be processed by GPU, (3) the periods of all time windows, and (4) the period and  number of remaining frames of each request. Essentially these four parts describe the existing system workload, and how to batch image frames to  job instances.

Once the Admission Control Module learns about the current system state, it proceeds to the second step, where it simulates the process of DisBatcher and generates pseudo job instances from all the requests, including the pending request being tested. 
It implements a virtual representation of the DisBatcher mechanism introduced in Section \ref{sec:model}, where the time and workload are both simulated. 
For each request, the Admission Control Module estimates the arrival time of each frame using its period, and compares the arrival time with the starting time and end time of the time windows to see which window each frame falls into. In this way the Admission Control Module is able to know the number of  frames to be batched in each window. The Admission Control Module looks up the execution time of each batched  job instance from the execution time table to get a list of virtual  job instances. This list contains the ``future'' job instances from all task instances, and the job instances are ranked according to their release times, which is the time when a job instance is pushed to the deadline queue. This is done through simultaneously running the aforementioned  DisBatcher simulator over all categories of requests and appending the  job instance with the smallest release time to the list each time.

With the system current state captured from Step 1 and the future virtual  job instances obtained from Step 2, the Admission Control Module moves to step three, where it  uses an EDF imitator algorithm to determine whether these job instances can be scheduled by EDF. The EDF imitator algorithm is shown in Algorithm $1$. In this algorithm, $Q$ is a sorted deadline queue storing    job instances, and $L$ is the list of job instances obtained from Step 2 which are released in the ``future''. Note that $L$ is already ordered by release times. We use $R_i$ to denote the release time of a job $i$, $E_i$ to denote the execution time of a job $i$, $D_i$ to denote the relative deadline of $i$, and $L[1]$ to denote the $1^{st}$ element of list $L$. The complexity of this algorithm is $O(N)$, where $N$ is the total number of frames.

\vspace{-12pt}
\subsection{DisBatcher Module and Execution Worker}

DisBatcher is the core component of {\name}. It is responsible for transforming  frames received from admitted user requests to job instances that are suitable to execute on GPU. It's an implementation of the batching approach presented in Section \ref{sec:model}.

\vspace{-0.1pt}
 Once a request is admitted by the Admission Control Module, the Admission Control Module sends request-related metadata including frame shape, requested model, period, and relative deadline to the DisBatcher. Then clients will directly send frames to DisBatcher. DisBatcher keeps track of all the admitted  requests in {\name}, together with their metadata.
DisBatcher manages a frame queue for each category of requests. These queues store  frames which arrive during their time windows and wait to be aggregated into batches of frames, or \emph{tensors}. DisBatcher utilizes reccurent countdown timers to implement time windows. It
 keeps a timer for each category of requests with the timer's countdown interval equal to the time window length. When a timer expires, DisBatcher batches all the  frames in its corresponding queue to a tensor and immediately starts timer countdown again. Whenever a new request is admitted, DisBatcher updates the countdown interval of the corresponding timer if the new request's relative deadline is smaller than the current smallest deadline.  

The DisBatcher wraps a tensor inside a new job instance, the relative deadline of which is equal to the corresponding time window length, and it pushes this  job instance onto an execution queue. A Worker on GPU subscribes to this execution queue and processes the job instances according to EDF. We implement the execution queue with a priority queue. The priority is determined by the job instances' absolute  deadlines (release time + relative deadline). 

The Worker is the execution engine which actually processes the batched job instances with a requested model on GPU.  It  repeatedly consumes the execution queue mentioned above whenever there are job instances inside. It processes the job instances one after another. The Worker is also responsible for monitoring the performance of execution. It  detects and records deadline misses. It also detects whether the execution time of a job instance is larger than the profiled worst-case execution time and report  overruns to the Adaptation Module.

We employ an optimization technique 
to further reduce frame latency. Occasionally GPU is idle while there are  frames which have already arrived but wait to be batched by DisBatcher. When {\name} detect this situation, it batches the frames before their timer expires and immediately sends the batched job instance to GPU for processing. In this way, {\name} can reduce the latency of these frames and  meanwhile increase the utilization of GPU.

\vspace{-5pt}
\subsection{Adaptation Module}

Since the hardware devices upon which {\name} operates are commercial off-the-shelf computing devices that don't have any hard deadline guarantees, {\name} is only able to provide soft real time services. The response time of processing the same request under the same setting may vary between different runs, occasionally leading to job overruns.  As the GPU executes jobs non-preemptively and we use EDF,  one single job overrun can possibly cause unpredictable deadline misses of many other jobs in the system. When overruns happen, we need a method to ``punish'' the overrun job, and more importantly, to avoid deadline misses as much as possible. 

In {\name}, each job instance category has a penalty initialized to be $0$. When the Worker observes that the execution time of a job instance exceeds the profiled execution time, the Adaptation Module will increase the penalty of the job instance category by the excess part. Meanwhile, the Adaptation Module informs the DisBatcher to decrease the shape (resolution)  of tensors belonging to that category. These tensors will not be batched with other tensors with the same smaller shape in order not to disturb job instances' priorities. The Worker will record the execution time of the new job instance and subtract the saved execution time from penalty. When penalty becomes non positive, the Adaptation Module commands the DisBatcher to resume the original shape of the overrun job instance and sets its penalty back to $0$. 

\ccomment{
\subsection{Supporting Non Real Time Requests}

{\name} also supports non real time requests, requests which don't require a latency guarantee. When non real time requests arrive, {\name} uses a polling server to transform them into periodically executed jobs. {\name} allocates a  portion of GPU utilization to these requests. The portion is by default set to be $20\%$ but it's configurable. }
\vspace{-2mm}
\vspace{-3pt}
\section{Implementation}
\label{sec:impl}



This section presents the implementation details of {\name}. All the scheduling actions of {\name} occur on CPU except the executions of CNN inference job instances. In order to make sure that {\name} gives prompt scheduling decisions, we assign the DisBatcher with the highest Linux user-space priority by setting its \texttt{nice} value.

We implement {\name} with the deep learning  framework Pytorch in Python. However, the mechanism of {\name} is completely framework agnostic. For communication between {\name} modules, we use a light-weighted messaging library ZeroMQ.

We use two edge devices equipped with GPUs to run and evaluate \name. The first device we use has a GeForce RTX 2080 Graphics Card with $2944$ CUDA cores. It has an Intel  i7-9700K 8-core CPU and $64$GB memory. The second device is an NVIDIA Jetson TX2 Developer Kit. Jetson TX2 is a computer specially designed to provide deep learning inference services on the edge. It is equipped with a GPU with $256$ CUDA cores. It has a 6-core CPU with $8$GB memory.

\vspace{-2mm}
\section{Evaluation}
\label{sec:eval}

We evaluate  {\name} by answering these questions:

\begin{itemize}
    \item How well does {\name} perform in terms of meeting deadline requirements as compared to state-of-the-art latency-centric CNN inference scheduling approaches?
    \item Is {\name} able to provide high throughput while guaranteeing soft real time services?
    \item How effective is the Admission Control Module in making schedulability decisions for new requests? What's the overhead of running this module?
    \item 
    How robust is {\name} against overruns and how quickly can {\name} bring the system back to normalcy? 

\end{itemize}

\vspace{-10pt}
\subsection{Experimental Dataset}

We use the DAVIS dataset \cite{perazzi2016benchmark} as the workload. This dataset consists of video frames of 480p ($480\times 854$)  and 1080p ($1080\times 1920$) resolution. We also downsample the video frames to various resolution formats in order to enrich the request data. On the desktop computer with an RTX 2080 card, we set the request video data to have 3 resolution formats: $1080\times 1920$, $480\times 854$, and $240\times 352$. On the Jetson TX2, as its computing power is smaller, we use frames of $360\times 640$, $240\times 352$, and $224\times 224$. All videos are colored videos and they all have the 3 RGB channels.
It is worth mentioning that {\name} is agnostic of video contents since different videos exhibit the same characteristics when being processed by classification models.

\vspace{-5pt}
\subsection{{\name} vs. Existing  Inference Systems }
\label{ssec:vs-state}

In this part, we compare the performance of {\name} and the state-of-the-art CNN inference systems with respect to their abilities to meet the latency requirements specified by user applications.

\textbf{Baseline.} We compare {\name} against these approaches which enable batching or adaptive batching to achieve low-latency high-throughput inference:

\begin{itemize}
    \item AIMD is an implementation of the dynamic batching scheme used by Clipper  \cite{crankshaw2017clipper} and MArk \cite{zhang2019mark}. As the name suggests, when inference latency does not exceed the latency objective, batch size increases additively. If latency objective is violated, a  multiplicative reduction of batch size is performed.
    \item BATCH is the scheme used by Triton Inference Server \cite{triton}. It performs batching over request data with a fixed batch size  determined empirically. We set the batch sizes as small as possible to reduce the latency of each job.
    \item BATCH-Delay is another scheme provided by Triton Inference Server. Apart from imposing a fixed number as the batch size, BATCH-Delay also imposes a time limit to each model. This scheme batches input data either when the number of frames in a batch reaches the configured batch size, or when the time limit is reached, whichever occurs first. 
\end{itemize}
It is worth mentioning that all of these approaches process multiple requests concurrently under multitenancy situation. For BATCH and BATCH-Delay, we set the  batch size as small as possible in order to reduce the latency of each job. However, sometimes a small batch size could drain GPU memory since there are too many jobs executed concurrently on GPU. If this happens, we increase the batch size until the memory problem is alleviated.

\textbf{Request traces.}  
Each time we run {\name} or the aforementioned inference systems, we feed the system with multiple synthesized requests to perform inference on their data. The requests are independent of each other and they arrive at the system one at a time. We use tweets traces from Twitter \cite{twitter} as a reference to determine the interval between the arrival of requests.
Each request contains a video with a fixed number of frames, and each frame  is released periodically according to its frame rate. In order to demonstrate the universality of {\name} on various kinds of applications, we randomly set the period and relative deadline of the frames in a video. The period and relative deadline of the frames are sampled from a Gamma distribution independently. We select Gamma distribution to generate random period and relative deadline settings because the generated random numbers by Gamma distribution start from $0$, and it's a common distribution in queuing theory.   The shape parameter $k$ and the scale parameter $\theta$ of the Gamma distribution are set to $2$ and $5$, respectively. Then we scale the samples to appropriate values. For each request, we randomly choose an input shape and a model from the $6$ models listed in Section \ref{sec:char} plus Mobilenet-v2 \cite{sandler2018mobilenetv2}, 
and we limit the number of categories of requests.

\begin{center}
\begin{table}
\footnotesize
\centering

\begin{tabular}{ c|c|c|c }
\hline
\hline
\multicolumn{4}{c}{Mean values of period and relative deadline (ms)} \\
\hline
   & Trace 1 & Trace 2 & Trace 3 \\ 
 \hline
Desktop & 50 & 150 & 250  \\ 
 \hline
 Jetson TX2 & 300 & 450 & 600\\
 \hline

\end{tabular}
\vspace{5pt}
\caption{Mean values of frame period and relative deadline when generating request traces.}
\label{tb:trace}
\vspace{-20pt}
\end{table}
\end{center}

\vspace{-15pt}
On both the desktop computer and Jetson TX2, we synthesize $3$ traces of requests using the aforementioned approach. Each trace contains $20$ to $30$ requests. The periods and relative deadlines of all requests in the same trace are obtained from scaling the randomly sampled values with the same factor. The mean values of periods and relative deadlines of these traces are shown in Table \ref{tb:trace}.

In each run, we feed an inference system with one trace of requests, wait till all frames are processed, and record frame deadline misses  as an indication of the system's ability to perform real time inference.  Since {\name}'s Admission Control Module admits requests selectively while other approaches don't have admission control measures, in order to guarantee the fairness of comparison, we record the accepted requests from {\name} and feed these requests to other systems. Moreover, we disable {\name}'s Adaptation Module which potentially reduces frame shapes.

\begin{figure*}
\vspace{-10pt}
  \centering
  \begin{subfigure}{0.45\textwidth}
  \centering
  \includegraphics[width = 0.65\linewidth]{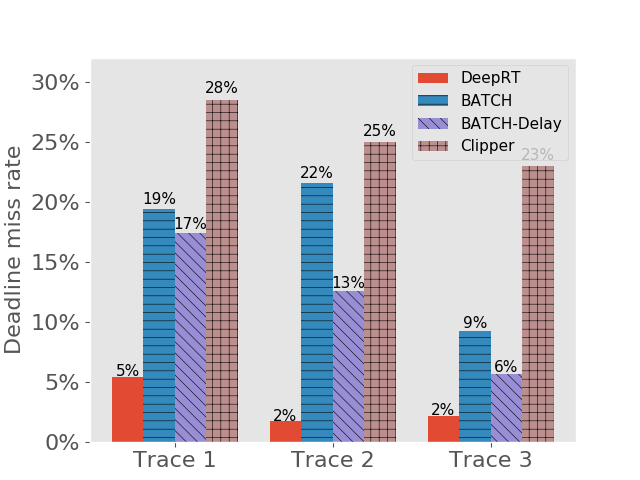}
  \caption{Desktop computer.}
  \label{fig:miss-rate-server}
  \end{subfigure}
  \begin{subfigure}{0.45\textwidth}
  \centering
  \includegraphics[width = 0.65\linewidth]{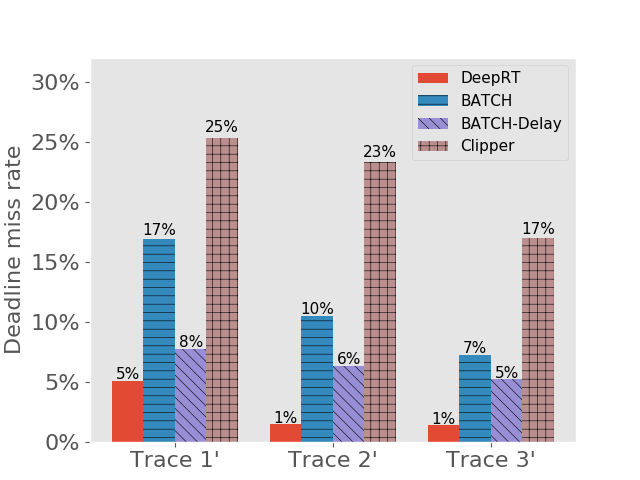}
  \caption{Jetson TX2.}
  \label{fig:miss-rate-jetson}
  \end{subfigure}
  \vspace{3pt}
  \caption{Comparison of deadline miss rates between {\name} and state-of-the-art inference systems on 3 synthesized request traces on the desktop computer and Jetson TX2.}
  \label{fig:miss-rate}
  \vspace{-13pt}
 \end{figure*}

 \begin{figure*}
  \centering
  \begin{subfigure}{0.32\textwidth}
  \includegraphics[width = 0.85\linewidth]{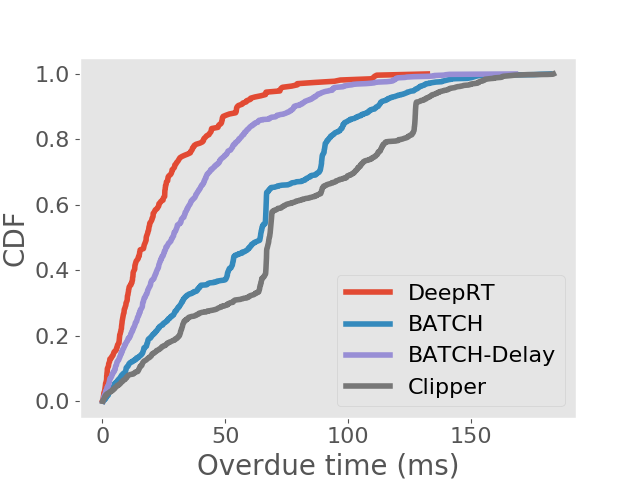}
  \caption{Trace 1 on desktop computer.}
  \label{fig:ddl-cdf-server-1}
  \end{subfigure}
  \begin{subfigure}{0.32\textwidth}
\includegraphics[width = 0.85\linewidth]{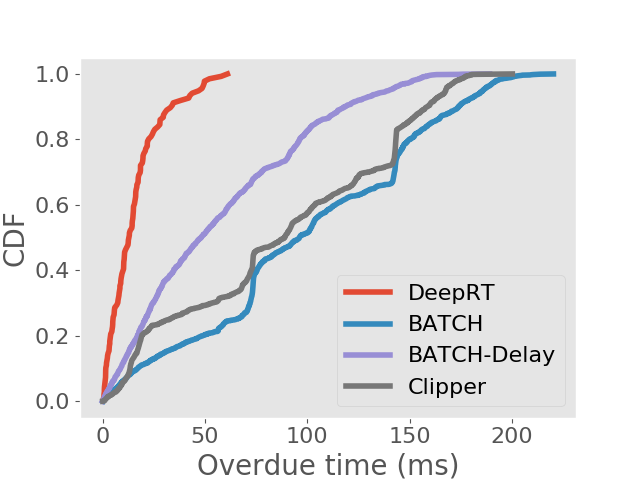}
  \caption{Trace 2 on desktop computer.}
  \label{fig:ddl-cdf-server-2}
  \end{subfigure}
\begin{subfigure}{0.32\textwidth}
  \includegraphics[width = 0.85\linewidth]{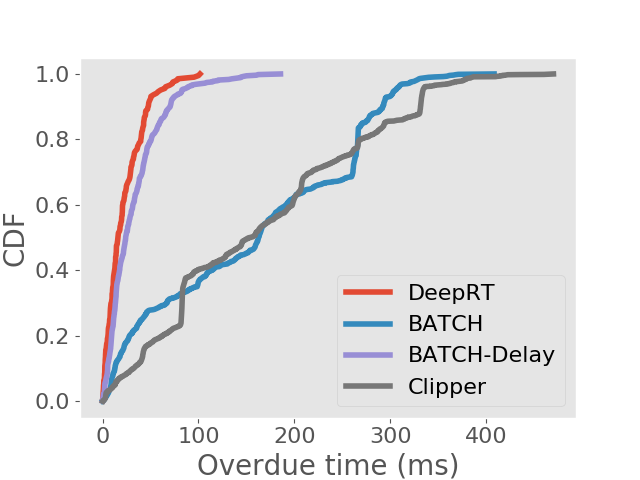}
  \caption{Trace 3 on desktop computer.}
  \label{fig:ddl-cdf-server-3}
  \end{subfigure}
  \begin{subfigure}{0.32\textwidth}
  \includegraphics[width = 0.85\linewidth]{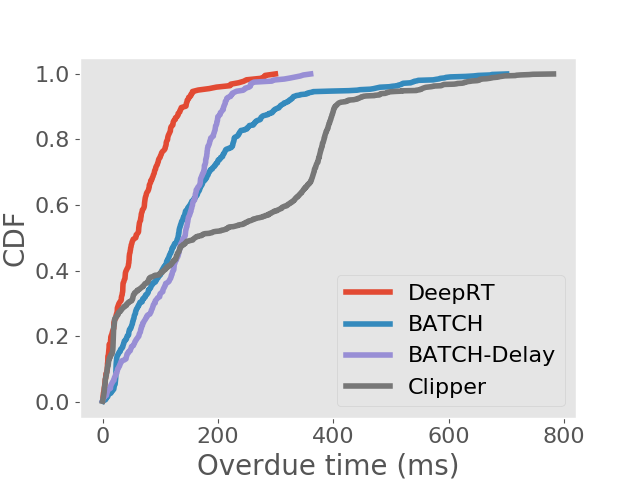}
  \caption{Trace 1' on Jetson TX2.}
  \label{fig:ddl-cdf-jetson-1}
  \end{subfigure}
  \begin{subfigure}{0.32\textwidth}
  \includegraphics[width = 0.85\linewidth]{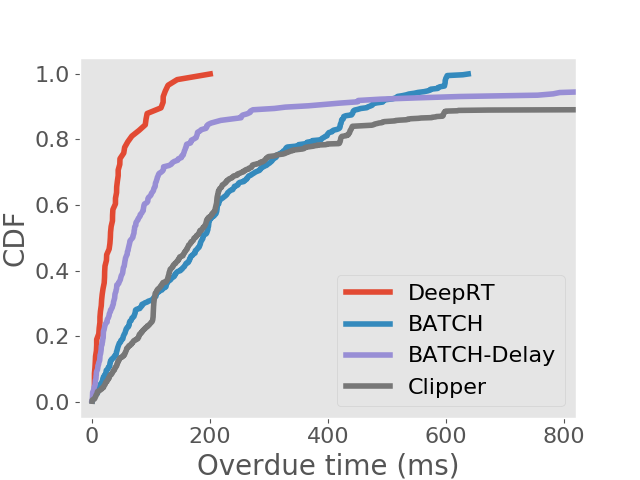}
  \caption{Trace 2' on Jetson TX2.}
  \label{fig:ddl-cdf-jetson-2}
  \end{subfigure}
  \begin{subfigure}{0.32\textwidth}
  \includegraphics[width = 0.85\linewidth]{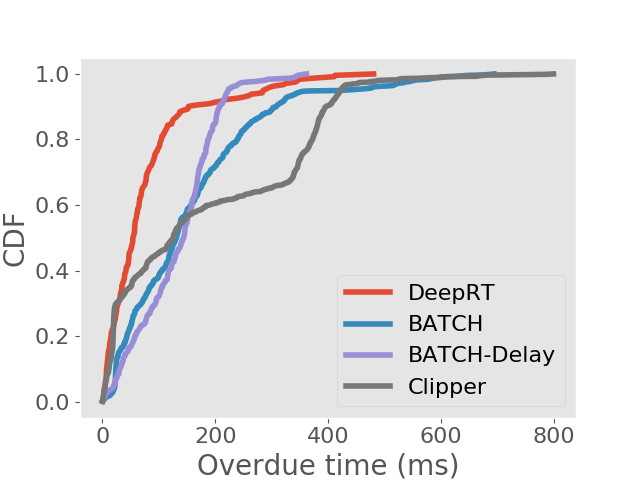}
  \caption{Trace 3' on Jetson TX2.}
  \label{fig:ddl-cdf-jetson-3}
  \end{subfigure}
  \vspace{3pt}
  \caption{CDF of overdue time under the synthesized request traces of Figure \ref{fig:miss-rate} on the desktop computer and Jetson TX2.}
  \label{fig:ddl-cdf}
  \vspace{-15pt}
 \end{figure*}

\textbf{Results on achieving real time inference.} We first show the deadline miss rates of {\name} and the other inference systems in Figure \ref{fig:miss-rate}. To distinguish  the traces on Jetson TX2 from the traces on desktop computer, we use ``Trace x$'$'' to indicate it is for Jetson TX2. We can see that for all $6$ traces {\name} shows the lowest deadline miss rates. When the mean values of period and relative deadline are $50ms$, the deadline miss rate of {\name} is still $5\%$ while it  handles $4$ concurrent requests and a total number of $10$ requests. {\name} exhibits the lowest deadline miss rates because its design focuses on meeting requests' deadlines and it considers special characteristics of GPU mentioned in Section \ref{sec:char}.  The results demonstrate {\name}'s ability to perform soft real time inference services for multiple requests. Note that {\name} does not completely avoid deadline misses due to job instance overruns.  Interestingly, Clipper shows the highest deadline miss rates in all runs. We think the reason is that the AIMD based adaptive batching scheme assumes abundant resources and is more suitable for cloud-scale inference.

\begin{figure}
  \centering
  \begin{subfigure}{0.23\textwidth}
  \includegraphics[width = 1.0\linewidth]{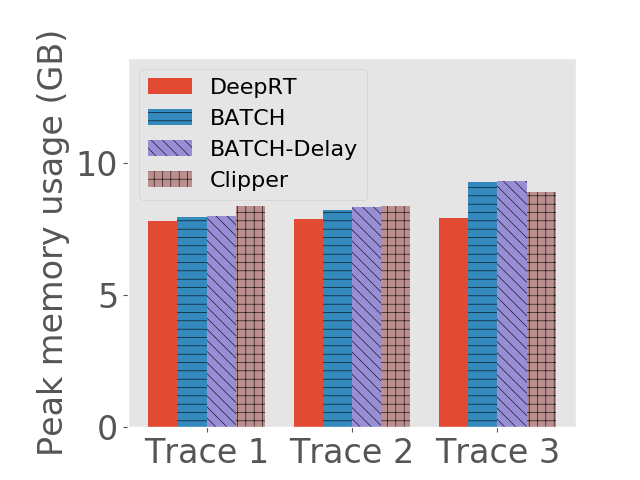}
  \caption{Desktop computer.}
  \label{fig:eval-memory-server}
  \end{subfigure}
  \begin{subfigure}{0.23\textwidth}
  \includegraphics[width = 1.0\linewidth]{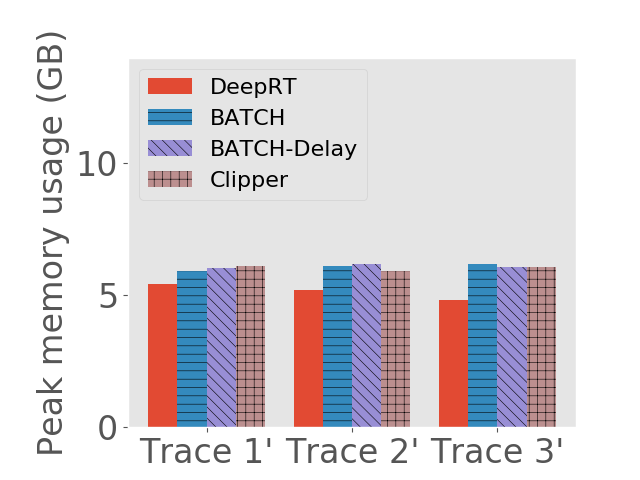}
  \caption{Jetson TX2.}
  \label{fig:eval-memory-jetson}
  \end{subfigure}
  \vspace{3pt}
  \caption{Peak memory usage of {\name} vs. state-of-the-art approaches under the request traces of Figure \ref{fig:miss-rate}.}
  \label{fig:eval-memory}
  \vspace{-15pt}
 \end{figure}
 
 \begin{figure*}
 \vspace{-10pt}
  \centering
  \begin{subfigure}{0.24\textwidth}
  \includegraphics[width = 0.95\linewidth]{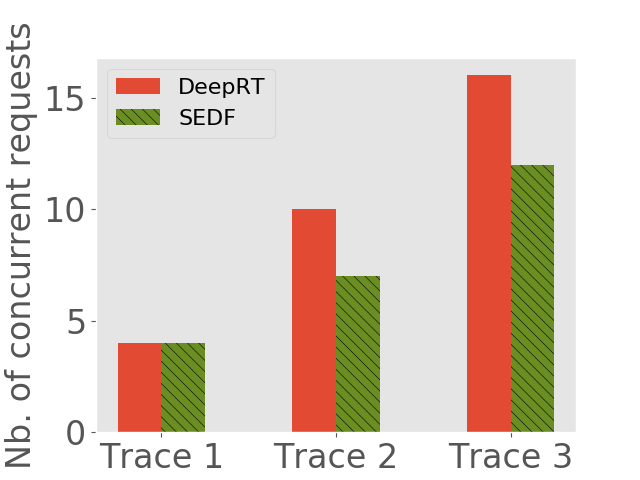}
  \caption{Number of admitted requests under 3 request traces on desktop computer.}
  \label{fig:req-server}
  \end{subfigure}
  \begin{subfigure}{0.24\textwidth}
  \includegraphics[width = 0.95\linewidth]{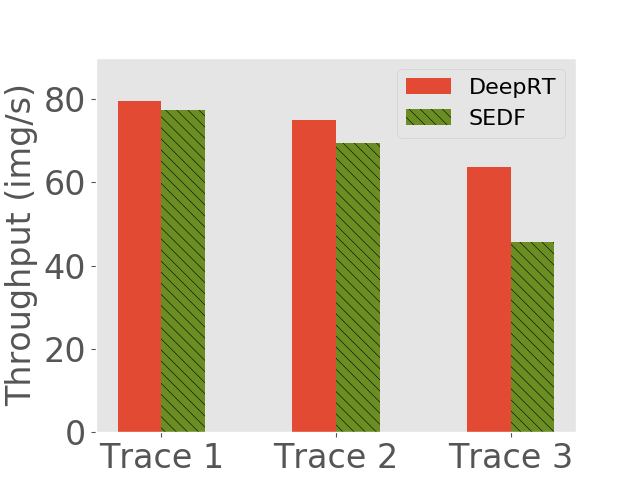}
  \caption{Throughput under 3 request traces on desktop computer.}
  \label{fig:thru-server}
  \end{subfigure}
    \begin{subfigure}{0.24\textwidth}
  \includegraphics[width = 0.95\linewidth]{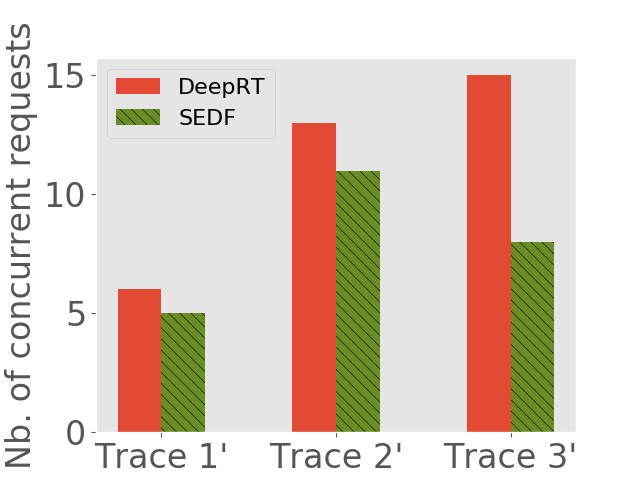}
  \caption{Number of admitted requests under 3 request traces on Jetson TX2.}
  \label{fig:req-jetson}
  \end{subfigure}
    \begin{subfigure}{0.24\textwidth}
  \includegraphics[width = 0.95\linewidth]{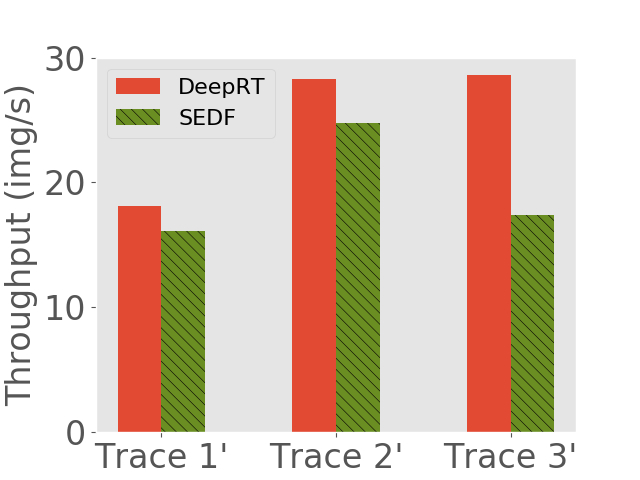}
  \caption{Throughput under 3 request traces on Jetson TX2.}
  \label{fig:thru-jetson}
  \end{subfigure}
  \vspace{3pt}
  \caption{Throughput comparison between {\name} and Sequential EDF on desktop computer and Jetson TX2.}
  \label{fig:throughput}
  \vspace{-15pt}
 \end{figure*}

We also analyse the frames that are finished processing after the deadlines. In soft real time systems, jobs whose deadlines are missed still provide some value if they can be completed as early as possible. We examine the distribution of overdue time for each inference approach and show the results in the form of CDFs in Figure \ref{fig:ddl-cdf}. We can see that {\name} performs the best in terms of deadline overdue time due to its utilization of the EDF algorithm.

\textbf{Peak GPU memory usage.} We also measure the peak GPU memory usage of the $4$ approaches under different traces. Peak GPU memory usage is an important metric since the computations on GPU are memory demanding. An inference system consuming too much GPU memory could drain the memory, leading to memory allocation errors. On the desktop computer we use \texttt{nvidia-smi} to measure GPU memory usage, and on Jetson TX2 we use \texttt{tegrastats}. The results are shown in Figure \ref{fig:eval-memory}.

\vspace{-10pt}
\subsection{Throughput performance of {\name}}

In this part we evaluate the throughput performance of {\name}. Since there is no existing real time scheduler for CNN inference on the edge, we implement a real time scheduler, Sequential EDF (SEDF), and compare the throughput performance of {\name} agaist SEDF. We would like to examine whether {\name} is able to offer high inference throughput while meeting latency requirements compared to SEDF. As its name suggests, SEDF processes input frames from multiple requests one by one according to the frames' deadlines. It doesn't execute multiple models concurrently, nor does it perform batching. We also implement an EDF imitator as the admission control policy of SEDF. It is worth mentioning that we don't compare {\name} with the baseline approaches in Section \ref{ssec:vs-state} since those approaches are not real time schedulers and therefore do not provide latency guarantees. Besides, they do not have admission control modules which reject requests if they can cause deadline misses. Therefore, we compare {\name} with a soft real time scheduler SEDF to guarantee a fair comparison.

We would like to see (1) how many concurrent requests  each inference system can handle and (2) what is the  throughput of each system. We use the same method as Subsection \ref{ssec:vs-state} to generate request traces, except that we increase the frequency of request arrival to saturate the inference systems. Another difference is that, we feed the two systems with the same pending requests, but we let each of them determine which requests to admit. The number of concurrent requests each approach can handle and the average throughput each approach achieves are shown in \ref{fig:throughput}.

We observe that on all the traces {\name} performs better than or as well as SEDF, due to the fact that the novel batching approach of {\name} leverages the batching ability of GPU and exhibits high throughput while providing latency guarantee. We can see that on the third traces of both devices, {\name} largely outperforms SEDF, while the differences on the first traces are smaller. As the mean  relative deadlines of the requests in the third traces are larger, more frames can be batched due to our DisBatcher design, boosting the throughput performance of {\name}. On the first traces, fewer frames are batched so {\name} doesn't have a high performance gain compared to SEDF.
 
 \vspace{-5pt}
 \subsection{Evaluating the Admission Control Module}
 
 In this part we evaluate the performance of the Admission Control Module. Specifically, we would like to examine (1) whether the Admission Control Module is able to accurately model the system to make admission decisions, and (2) what is the running time of the Admission Control Module.
 
\textbf{Accuracy of the EDF imitator.} Naturally, we would like the Admission Control Module to admit as many requests as possible  to increase  throughput while not violating any deadline requirements. That is the reason why we employ an EDF imitator as an exact analysis tool to determine schedulability. We evaluate how accurate the EDF imitator is in estimating future job instance executions. 

We generate $3$ traces of requests with the method  in Subsection \ref{ssec:vs-state}. The only difference lies in the mean values of periods and relative deadlines. For the first trace, we set the mean period to be $100ms$ and mean relative deadline to be $300ms$. For the second trace, we set both values to be $200ms$. And for the third trace we set them to be $300ms$ and $100ms$, respectively. The reason of using these configurations is that we would like to examine the EDF imitator under various batch sizes and various deadlines.
We only perform this experiment on the desktop computer as the effectiveness of the EDF imitator is the same across all platforms as long as the profiled worst-case job instance execution times are accurate.
\ccomment{
\begin{figure}
  \centering
  \includegraphics[width = 0.6\linewidth]{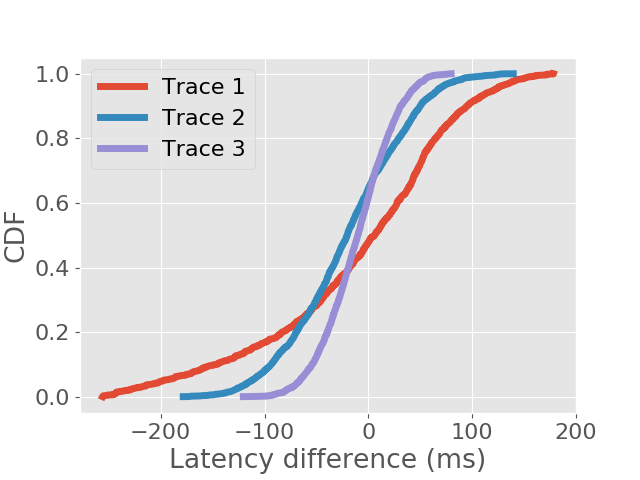}
  \vspace{3pt}
  \caption{CDF of the differences between predicted latencies by the EDF imitator and  latencies measured in real executions.}
  \label{fig:eval-accuracy}
  \vspace{-15pt}
 \end{figure}
 
 \begin{figure}
  \centering
  \includegraphics[width = 0.6\linewidth]{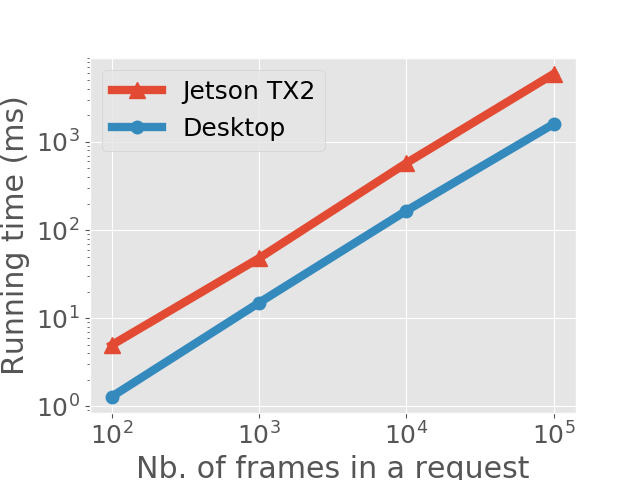}
  \vspace{3pt}
  \caption{Running time of the Admission Control Module when the requests contain different number of frames.}
  \label{fig:eval-run-time}
  \vspace{-15pt}
 \end{figure}
}

\begin{figure}
  \centering
  \begin{minipage}{.485\linewidth}
  \centering
  \includegraphics[width = 1.\linewidth]{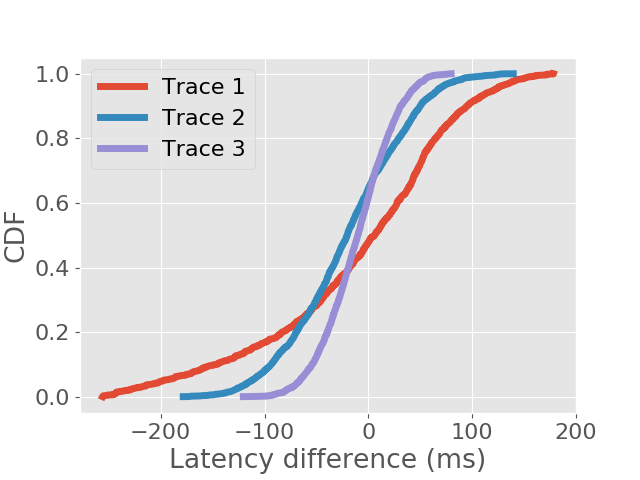}
  \caption{CDF of the differences  between predicted latencies by the EDF imitator and  latencies measured in real executions.}
  \label{fig:eval-accuracy}
  \end{minipage}%
  \hspace{4pt}
    \begin{minipage}{.485\linewidth}
  \centering
  \includegraphics[width = 1.\linewidth]{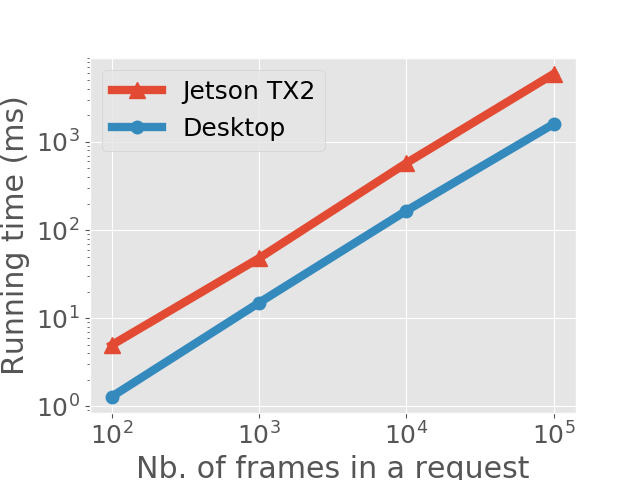}
  \caption{Median running time of the Admission Control Module when the requests contain different number of frames.}
  \label{fig:eval-run-time}
  \end{minipage}%
  \vspace{-20pt}
 \end{figure}

We use the  difference between the estimated latency of frame inference from the EDF imitator and the actual frame latency measured during real executions as the metric of accuracy, as the major goal of {\name} is to provide latency guarantee for users. The result is shown in Figure \ref{fig:eval-accuracy}. We can see that the difference is the smallest on the trace with the smallest deadline, and \emph{vice versa}. On the first trace which corresponds to the largest relative deadline ($300ms$), the difference can be as large as $250ms$. We find that the large latency differences happen on latter frames in a request frame sequence. In fact, when  the EDF imitator is performed on some requests, it considers all frames in the requests and latency estimation errors accumulate over the frame sequences. But large latency differences are rare and still smaller than the corresponding relative deadlines ($250ms<300ms$).  Overall the EDF imitator is sufficiently accurate to predict whether the deadline of a frame will be missed.

\textbf{Admission Control Module running time.} 
As mentioned in Section \ref{sec:design}, the complexity of the EDF imitator is linear with respect to the number of frames in all requests. We evaluate the running time of the Admission Control Module on both devices under different number of frames. Specifically, we generate $4$ request traces for both devices using the previous method; the requests in the $4$ traces contain videos with $10^2$, $10^3$, $10^4$, and $10^5$ frames, respectively.  The  running times (see Figure \ref{fig:eval-run-time}) are all below $1$ second except the case where Jetson TX2 processes requests with $10^5$ frames, where the running time  is $5.9$ seconds. If we consider the normal frame rate of a video to be $30fps$, $10^5$ frames correspond to a video of approximately one hour. In fact, if {\name} is used to perform inference on long videos, we can calculate the least common multiple of their periods and run the EDF imitator over twice that time period, significantly reducing the running time.

 \vspace{-5pt}
 \subsection{Adapt to Overruns}
 
 We evaluate how quickly {\name} reacts to job instance overruns and eliminate deadline misses caused by these overruns. We generate request traces with periods and relative deadlines to be $200ms$ for desktop and $600ms$ for Jetson TX2. In each run we manually inject a short waiting time to $5$ consecutive job instances and measure the number of deadline misses caused by the injected waiting time. If a certain method reacts to overruns faster, it can reduce the number of deadline misses.  We run this experiment on both the desktop computer and Jetson TX2. The lengths of the waiting times are set to be $100ms$, $200ms$, $500ms$, and $1000ms$.
 
 We compare the number of deadline misses between enabling and disabling the Adaptation Module in Figure \ref{fig:eval-adapt}. We can see that even without the Adaptation Module, {\name} is still able to bring the system back to normalcy after experiencing some deadline misses. The reason is that {\name} does not achieve $100\%$ utilization of the GPU as it is a  real time system. There is idle time between job instance executions which act as buffer absorbing the overruns. The Adaptation Module enhances the ability to absorb the overruns.

\vspace{-2mm}
\section{Related Work}
\label{sec:relatedwork}

\textbf{\ \ \  Deep Learning Inference on the Edge.}  
There has been a surge in industrial and research efforts to develop deep learning inference systems  on cloud or on edge. Tensorflow-Serving \cite{olston2017tensorflow} and Triton Inference Server \cite{triton} are two industrial general-purpose inference platforms. Clipper \cite{crankshaw2017clipper}  is a cloud based  throughput and latency oriented model serving system with a modular design. Mainstream \cite{jiang2018mainstream} enables work sharing among different vision applications to increase  throughput. In \cite{fang2017qos}, the authors propose to use deep reinforcement learning to adaptively select model and batch size to optimize QoS  defined as combination of accuracy and latency. Swayam \cite{gujarati2017swayam} is a cloud based machine learning serving system which autoscales resources  to meet SLA goals. In \cite{zhou2019adaptive}, the authors propose to partition CNN models across multiple IoT devices  to speed up the inference. Most of the above  efforts aim to improve throughput and latency performance, but either (1) they assume abundant cloud resources and achieve their performance goals through scaling the resources, or (2) they improve latency performance but don't have soft real time guarantee. DeepQuery \cite{fang2019serving} co-schedules real time and non-real time tasks on GPU, but its primary focus is to optimize the performance for non-real time tasks.

\begin{figure}
  \centering
  \begin{subfigure}{0.23\textwidth}
  \includegraphics[width = 1.0\linewidth]{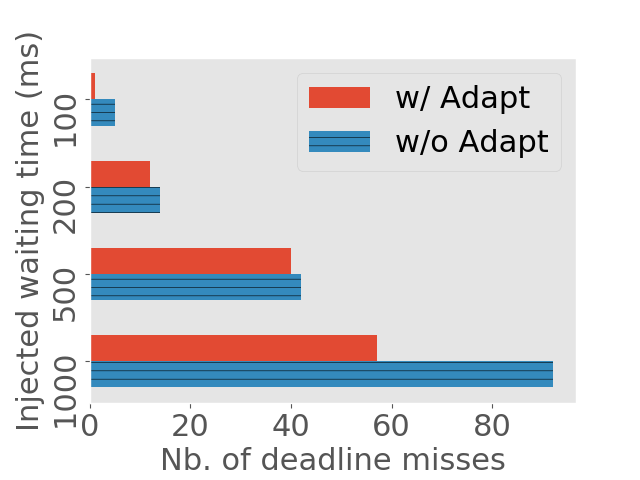}
  \caption{On desktop computer.}
  \label{fig:eval-adapt-server}
  \end{subfigure}
  \begin{subfigure}{0.23\textwidth}
  \includegraphics[width = 1.0\linewidth]{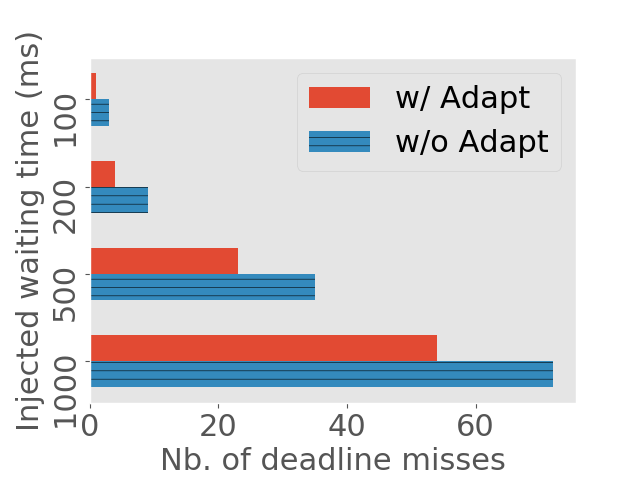}
  \caption{On Jetson TX2.}
  \label{fig:eval-adapt-jetson}
  \end{subfigure}
  \vspace{3pt}
  \caption{Comparison of the number of deadline misses caused by manually injected overruns between enabling and disabling the Adaptation Module.}
  \label{fig:eval-adapt}
  \vspace{-15pt}
 \end{figure}

Meanwhile a few works study the performance characteristics of deep learning inference on the edge. In \cite{hanhirova2018latency} the authors study latency and throughput performance of object recognition and object detection deep learning models. Their study focuses more on latency-throughput trade-off due to different batch sizes. In \cite{liang2020ai} the authors characterized different  AI applications using their specially built edge accelerators, and evaluate the benefit of model splitting.

\textbf{Augmenting GPU with Real Time Features.} 
Some researchers look deeper into how GPU works and try to add real time features to GPU based computations. 
In \cite{wang2017quality}, the authors propose to provide QoS support for GPU applications through fine-grained management of GPU resources such as registers, memory, and computation cycles. GPUSync \cite{elliott2013gpusync} is a real time management framework  supporting multiple scheduling policies such as rate-monotonic and EDF  using synchronization-based management. 
In \cite{liu2020removing}, the authors propose to separate CNN input data into different regions of importance and prioritize critical tasks by optimizing the importance of  regions. There have also been some works providing GPU with preemption ability by implementing GPU context switches \cite{tanasic2014enabling}\cite{park2015chimera}\cite{wang2016simultaneous}\cite{wu2017flep}. All these works differ from our work in that they provide real time features to GPU processing by manipulating lower level components such as GPU driver.

\textbf{Latency-centric IoT Data Processing.} 
Apart from processing computer vision application requests using GPU, researchers have proposed various  scheduling systems to perform traditional processing on video or IoT contents. 
In \cite{chu1999cpu}, the authors propose a CPU service class for multimedia real time processing, and put forward some scheduling  algorithms to process different service classes on CPU.  Janus \cite{rivas2010janus} provides a cross-layer CPU scheduling architecture for virtual machine monitors to schedule soft real time multimedia processing tasks. VideoStorm \cite{zhang2017live} has an offline profiler which generates videos' resource-quality profile, and uses this profile to jointly optimize processing quality and latency. Miras \cite{yang2019miras} proposes a reinforcement learning based scheduling scheme for scientific workflow data on cloud, minimizing average response time of the workflow requests.
\vspace{-2mm}
\section{Conclusion}
\label{sec:conc}

We present {\name}, a soft real time scheduler for performing CNN inference on the edge. {\name} consists of $5$ modules -- a Performance Profiler, an Admission Control Module, DisBatcher, an Execution Worker, and an Adaptation Module. {\name} uses time windows to  batch input data, the lengths of which are determined by the requests' deadlines, and processes the batched data sequentially. Our evaluation results show that {\name} is able to provide guarantee on inference latency while maintaining high inference throughput. 



%
\begin{acks}
This work is supported by the National Science Foundation under grant NSF 1827126.
\end{acks}

{\footnotesize \bibliographystyle{acm}}
\balance
{
\small
\bibliography{egbib}
}

\end{document}